\documentclass[UKenglish,cleveref,numberwithinsect,thm-restate]{no-lipics-v2019}

\RequirePackage{doi}
\usepackage{hyperref}
\usepackage{blkarray}
\usepackage{bm}

\newcommand{\NP}{\ensuremath{\mathsf{NP}}}
\newcommand{\W}[1]{\ensuremath{\mathsf{W[#1]}}}

\newcommand{\homs}[2]{\mbox{\ensuremath{\mathsf{Hom}(#1 \to #2)}}}

\newcommand{\indsubs}[2]{\mbox{\ensuremath{\mathsf{IndSub}(#1 \to #2)}}}
\newcommand{\auts}[1]{\ensuremath{\mathsf{Aut}(#1)}}

\newcommand{\homsprob}{\ensuremath{\textsc{Hom}}}

\newcommand{\indsubsprob}{\ensuremath{\textsc{IndSub}}}

\newcommand{\fpol}{\ensuremath{\mathtt{f}}}
\newcommand{\lovasz}{Lov{\'{a}}sz}

\title{Counting Small Induced Subgraphs\texorpdfstring{\\}{ }Satisfying Monotone Properties}{}
\titlerunning{Counting Small Induced Subgraphs Satisfying Monotone Properties}

\author{Marc Roth}{Merton College, University of Oxford, United Kingdom}{marc.roth@merton.ox.ac.uk}{https://orcid.org/0000-0003-3159-9418}{}

\author{Johannes Schmitt}{Mathematical Institute, University of Bonn, Germany}{schmitt@math.uni-bonn.de}{https://orcid.org/0000-0001-5774-3508}{}

\author{Philip Wellnitz}{Max Planck Institute for Informatics,\and Saarland Informatics Campus
(SIC), Saarbrücken,
Germany}{wellnitz@mpi-inf.mpg.de}{https://orcid.org/0000-0002-6482-8478}{}

\authorrunning{M. Roth, J. Schmitt, and P. Wellnitz}

\Copyright{Marc Roth, Johannes Schmitt, and Philip Wellnitz}


\nolinenumbers 

\begin{document}

\maketitle

\begin{abstract}
    Given a graph property $\Phi$, the problem $\#\indsubsprob(\Phi)$ asks, on input a
    graph~$G$ and a positive integer $k$, to compute the number  $\#\indsubs{\Phi,k}{G}$
    of induced subgraphs of size $k$ in $G$ that satisfy $\Phi$.
    The search for  \emph{explicit} criteria on $\Phi$ ensuring that
    $\#\indsubsprob(\Phi)$ is hard was initiated by Jerrum and Meeks~[J.\ Comput.\
    Syst.\ Sci.\ 15] and is part of the major line of research on counting small patterns in graphs.
    However, apart from an implicit result due to Curticapean, Dell and
    Marx [STOC 17] proving that a full classification into ``easy'' and ``hard''
    properties is possible and some partial
    results on edge-monotone properties due to Meeks [Discret.\ Appl.\ Math.\ 16] and
    Dörfler et al.\ [MFCS~19], not much is known.

    In this work, we fully answer and explicitly classify the case of monotone, that is
    subgraph-closed, properties: We show that for any non-trivial monotone property
    $\Phi$, the problem $\#\indsubsprob(\Phi)$ cannot be solved in time $f(k)\cdot
    |V(G)|^{o(k/ {\log^{1/2}(k)})}$ for any function $f$, unless the Exponential Time
    Hypothesis fails. By this, we establish that any significant improvement over the
    brute-force approach is unlikely; in the language of parameterized complexity, we also
    obtain a $\#\W{1}$-completeness result.

    To prove our result, we use that for fixed $\Phi$ and $k$, we can express the function
    $G \mapsto \#\indsubs{\Phi, k}{G}$ as a finite linear-combination of homomorphism
    counts from graphs $H_i$ to $G$.
    The coefficient vectors of these homomorphism
    counts in the linear combination are called the \emph{homomorphism vectors} associated
    to $\Phi$; by the Complexity Monotonicity framework of Curticapean, Dell
    and Marx [STOC 17], the positions of non-zero entries of these vectors are known to
    determine the complexity of $\#\indsubsprob(\Phi)$.
    Our main technical result lifts the notion of $f$-polynomials from
    simplicial complexes to graph properties and relates the derivatives
    of the $f$-polynomial of~$\Phi$ to its homomorphism vector. We then apply
    results from the theory of Hermite-Birkhoff interpolation to the
    $f$-polynomial to establish sufficient conditions on $\Phi$ which ensure
    that certain entries in the homomorphism vector do not vanish---which in turn implies
    hardness.
    For monotone graph properties, non-triviality then turns out to be a sufficient condition.
    Using the same method, we also prove a conjecture by Jerrum and Meeks [TOCT\,15,
    Combinatorica\,19]: $\#\indsubsprob(\Phi)$ is $\#\W{1}$-complete if
    $\Phi$ is a non-trivial graph property only depending on the number of edges of the
    graph.
\end{abstract}

\newpage

\section{Introduction}
Detection, enumeration and counting of patterns in graphs are among the most well-studied
computational problems in theoretical computer science with a plethora of applications in
diverse disciplines, including biology~\cite{SchreiberS05,GrochovK07}, statistical
physics~\cite{TemperleyF61,Kasteleyn61,Kasteleyn63}, neural and social
networks~\cite{Miloetal02} and database theory~\cite{GroheSS01}, to name but a few. At the
same time, those problems subsume in their unrestricted forms some of the most infamous
$\mathrm{NP}$-hard problems such as Hamiltonicity, the clique problem, or, more generally,
the subgraph isomorphism problem~\cite{Cook71,Ullmann76}. In the modern-day era of ``big
data'', where even quadratic-time algorithms may count as inefficient, it is hence crucial
to find relaxations of hard computational problems that allow for tractable instances.

A very successful approach for a more fine-grained understanding of hard computational
problems is a multivariate analysis of the complexity of the problem: Instead of
establishing upper and (conditional) lower bounds only depending on the input size, we aim
to find additional parameters that, in the best case, are small in real-world instances
and allow for efficient algorithms if assumed to be bounded. In case of detection and
counting of patterns in graphs, it turns out that the size of the pattern is often
significantly smaller than the size of the graph: Consider as an example
the evaluation of database queries. While a classical analysis of this problem
requires considering instances where the size of the query is as large as the database, a
multivariate analysis allows us to impose the restriction of the query being much smaller
than the database, which is the case for real-world instances. More concretely, suppose we
are given a query $\varphi$ of size $k$ and a database $B$ of size $n$, and we wish to
evaluate the query $\varphi$ on $B$. Assume further, that we are given two algorithms for
the problem: One has a running time of $O(n^k)$, and the other one has a running time of
$O(2^k \cdot n)$. While, classically, both algorithms are inefficient in the sense that
their running times are not bounded by a polynomial in the input size $n+k$, the second
algorithm is significantly better than the first one for real-world instances and can even
be considered efficient.

In this work, we focus on \emph{counting} of small patterns in large graphs. The field of
counting complexity was founded by Valiant's seminal result on the complexity of computing
the permanent~\cite{Valiant79,Valiant79b}, where it was shown that computing the number of
perfect matchings in a graph is $\#\mathrm{P}$-complete, and thus harder than every
problem in the polynomial-time hierarchy $\mathrm{PH}$~\cite{Toda91}. This is in sharp
contrast to the fact that \emph{finding} a perfect matching in a graph can be done in
polynomial-time~\cite{Edmonds65}. Hence, a perfect matching is a pattern that allows for
efficient detection but is unlikely to admit efficient counting. Initiated by Valiant,
computational counting evolved into a well-studied subfield of theoretical computer
science. In particular, it turns out that counting problems are closely related to
computing partition functions in statistical
physics~\cite{TemperleyF61,Kasteleyn61,Kasteleyn63,GoldbergJ19,Chenetal19,Backensetal20}.
Indeed, one of the first algorithmic result in the field of computational counting is the
famous FKT-Algorithm by the statistical physicists Fisher, Kasteleyn and
Temperley~\cite{TemperleyF61,Kasteleyn61,Kasteleyn63} that computes the partition function
of the so-called dimer model on planar structures, which is essentially equivalent to
computing the number of perfect matchings in a planar graph. The FKT-Algorithm is the
foundation of the framework of holographic algorithms, which, among others, have been used
to identify the tractable cases of a variety of complexity classifications for counting
constraint satisfaction
problems~\cite{Valiant08,CaiL11,CaiFGW15,CaiHL12,HuangL16,CaiLX17,CaiLX18,Backens18}.
Unfortunately, the intractable cases of those classifications indicate that, except for
rare examples, counting is incredibly hard (from a complexity theory point of view).
In particular, many efficiently solvable combinatorial
decision problems turn out to be intractable in their counting versions, such as
counting of satisfying assignments of monotone $2$-CNFs~\cite{Valiant79b}, counting of
independent sets in bipartite graphs~\cite{ProvanB83} or counting of
$s$-$t$-paths~\cite{Valiant79b}, to name but a few. For this reason, we follow the
multivariate approach as outlined previously and restrict ourselves in this work on
counting of \emph{small} patterns in large graphs. Among others, problems of this kind
find applications in neural and social networks~\cite{Miloetal02}, computational
biology~\cite{Nogaetat08,Schilleretal15}, and database
theory~\cite{DurandM15,ChenM15,ChenM16,DellRW19icalp}.
\newpage

\noindent Formally, we follow the approach of Jerrum and Meeks~\cite{JerrumM15} and study
the family of problems $\#\indsubsprob(\Phi)$: Given a graph property $\Phi$, the problem
$\#\indsubsprob(\Phi)$ asks, on input a graph $G$ and a positive integer $k$, to compute
the number of induced subgraphs of size $k$ in $G$ that satisfy $\Phi$.\footnote{Note that
$\#\indsubsprob(\Phi)$ is identical to
$p$-$\#\textsc{UnlabelledInducedSubgraphWithProperty}(\Phi)$ as defined
in~\cite{JerrumM15}.} As observed by Jerrum and Meeks, the generality of the definition
allows to express counting of almost arbitrary patterns of size $k$ in a graph, subsuming
counting of $k$-cliques and $k$-independent sets as very special cases.

Assuming that $\Phi$ is computable, we note that the problem $\#\indsubsprob(\Phi)$ can be
solved by brute-force in time $O(f(k)\cdot |V(G)|^k)$ for some function $f$ only depending
on $\Phi$. The corresponding algorithm enumerates all subsets of $k$ vertices of $G$ and
counts how many of those subsets satisfy $\Phi$. As we consider $k$ to be significantly
smaller than $|V(G)|$, we are interested in the dependence of the exponent on~$k$. More
precisely, the goal is to find the best possible $g(k)$ such that $\#\indsubsprob(\Phi)$
can be solved in time
\begin{equation}\label{eq:intro_goal}
O(f(k)\cdot |V(G)|^{g(k)})
\end{equation}
for some function $f$ such that $f$ and $g$ only depend on $\Phi$. Readers familiar with
parameterized complexity theory will identify the case of $g(k)\in O(1)$ as
fixed-parameter tractability (FPT) results. We first provide some background and
elaborate on the existing results on $\#\indsubsprob(\Phi)$ before we present the
contributions of this paper.

\subsection{Prior Work}
So far, the problem $\#\indsubsprob(\Phi)$ has been investigated using primarily the framework of
parameterized complexity theory. As indicated before, $\#\indsubsprob(\Phi)$ is in
$\mathrm{FPT}$ if the function $g$ in \cref{eq:intro_goal} is bounded by a
constant (independent of $k$), and the problem is $\#\W{1}$-complete, if it is at least as
hard as the parameterized clique problem; here $\#\W{1}$ should be considered a
parameterized counting equivalent of $\mathrm{NP}$ and we provide the formal details in
\cref{sec:prelims}. In particular, the so-called Exponential Time Hypothesis (ETH) implies
that $\#\W{1}$-complete problems are not in $\mathrm{FPT}$; again, this is made formal in
\cref{sec:prelims}.

The problem $\#\indsubsprob(\Phi)$ was first studied by Jerrum and Meeks~\cite{JerrumM15}.
They introduced the problem and proved that $\#\indsubsprob(\Phi)$ is $\#\W{1}$-complete
if $\Phi$ is the property of being connected. Implicitly, their proof also rules out the
function $g$ of \cref{eq:intro_goal} being in $o(k)$, unless ETH fails, which
establishes a tight conditional lower bound.
In a subsequent line of research~\cite{JerrumM15density,Meeks16,JerrumM17}, Jerrum and
Meeks proved $\#\indsubsprob(\Phi)$ to be $\#\W{1}$-complete if at least one of the following
is true:
\begin{enumerate}[(1)]
    \item $\Phi$ has \emph{low edge-densities}; this is true for instance for all sparse
        properties such as planarity and made formal in \cref{sec:low_edge_dens}.
	\item $\Phi$ holds for a graph $H$ if and only if the number of edges of $H$ is even/odd.
    \item $\Phi$ is closed under the addition of edges, and the minimal elements have large
        treewidth.
\end{enumerate}
Unfortunately, none of the previous results establishes a conditional lower bound that
comes close to the upper bound given by the brute-force algorithm. This is particularly true
due to the application of Ramsey's Theorem in the proofs of many of the prior results:
Ramsey's Theorem states that there is a function $R(k) = 2^{\Theta(k)}$ such that
every graph with at least $R(k)$ vertices contains either a $k$-independent set or a
$k$-clique~\cite{Ramsey30,Spencer75,Erdos87}. Relying on this result for a reduction from
finding or counting $k$-independent sets or $k$-cliques, the best implicit conditional
lower bounds achieved only rule out an algorithm running in time
$f(k)\cdot |V(G)|^{o(\log k)}$ for any function $f$.

Moreover, the previous results only apply to a very
specific set of properties. In particular, Jerrum and Meeks posed the following open problem
concerning a generalization of the second result (2); we say that~$\Phi$ is $k$-trivial,
if it is either true or false for all graphs with $k$ vertices.
\begin{conjecture}[\cite{JerrumM15density,JerrumM17}]\label{conj:JM}
    Let $\Phi$ be a graph property that only depends on the number of edges of a graph. If
    for infinitely many $k$ the property $\Phi$ is not $k$-trivial, then $\#\indsubsprob(\Phi)$ is
    $\#\W{1}$-complete.\lipicsEnd
\end{conjecture}
Note that the condition of $\Phi$ not being $k$-trivial for infinitely many $k$ is
necessary for hardness, as otherwise, the problem becomes trivial if $k$ exceeds a
constant depending only on $\Phi$.

The first major breakthrough towards a complete understanding of the complexity of
$\#\indsubsprob(\Phi)$ is the following implicit classification due to Curticapean, Dell
and Marx~\cite{CurticapeanDM17}:
\begin{theorem}[\cite{CurticapeanDM17}]
    Let $\Phi$ denote a graph property. Then the problem $\#\indsubsprob(\Phi)$ is either
    $\mathrm{FPT}$ or $\#\W{1}$-complete.\lipicsEnd
\end{theorem}
While the previous classification provides a very strong result for the structural
complexity of $\#\indsubsprob(\Phi)$, it leaves open the question of the precise bound on
the function $g$. Furthermore, it is implicit in the sense that it does not reveal the
complexity of $\#\indsubsprob(\Phi)$ if a concrete property $\Phi$ is given. Nevertheless, the
technique introduced by Curticapean, Dell and Marx, which is now called \emph{Complexity
Monotonicity}, turned out to be the right approach for the treatment of
$\#\indsubsprob(\Phi)$. In particular, the subsequent results on $\#\indsubsprob(\Phi)$,
including the classification in this work, have been obtained by strong refinements of
\emph{Complexity Monotonicity}; we provide a brief introduction when we discuss the
techniques used in this paper.
More concretely, a superset of the authors established the following classifications for
edge-monotone properties in recent years~\cite{RothS18,DorflerRSW19}: The problem
$\#\indsubsprob(\Phi)$ is $\#\W{1}$-complete and, assuming ETH, cannot be solved in time
$f(k)\cdot |V(G)|^{o(k)}$ for any function $f$, if at least one of the following is
true:\footnote{We provide simplified statements; the formal and more general results can
be found in~\cite{RothS18,DorflerRSW19}.}
\begin{itemize}
    \item $\Phi$ is non-trivial, closed under the removal of edges and false on odd
        cycles.
	\item $\Phi$ is non-trivial on bipartite graphs and closed under the removal of edges.
\end{itemize}
While the second result completely answers the case of edge-monotone properties on
bipartite graphs, a general classification of edge-monotone properties is still unknown.

\subsection{Our Results}
We begin with monotone properties, that is, properties that are closed under taking
subgraphs. We classify those properties completely and explicitly; the following theorem
establishes hardness and an almost tight conditional lower bound.
\begin{mtheorem}\label{thm:monotone_refined_intro}
    Let $\Phi$ denote a monotone graph property.
    Suppose that for infinitely many~$k$ the property $\Phi$ is not $k$-trivial. Then
    $\#\indsubsprob(\Phi)$ is $\#\W{1}$-complete and cannot be solved in time
	$f(k)\cdot |V(G)|^{o(k/\sqrt{\log k})}$ for any function $f$, unless ETH fails.
    \lipicsEnd
\end{mtheorem}
In fact, we obtain a tight bound, that is, we can drop the factor of $1/\sqrt{\log
k}$ in the exponent, assuming the conjecture that ``you cannot beat treewidth''~\cite{Marx10}. The latter is an
important open problem in parameterized and fine-grained complexity theory asking for a
tight conditional lower bound for the problem of \emph{finding} a homomorphism from a
small graph $H$ to a large graph $G$: The best known algorithm for that problem runs in
time $\mathsf{poly}(|V(H)|) \cdot |V(G)|^{O(\mathsf{tw}(H))}$, where $\mathsf{tw}(H)$ is
the treewidth\footnote{We will only rely on treewidth in a black-box manner in this paper and thus refer the reader for instance to \cite[Chapter~7]{CyganFKLMPPS15} for a detailed treatment.} of $H$ (see for instance \cite{DiazST02,Marx10,CurticapeanDM17}), and the question is whether this running time is essentially optimal; we discuss the details later in the paper.

\noindent As a concrete example of a property that is classified by
\cref{thm:monotone_refined_intro}, but that was not classified before,
consider the (monotone) property of being $3$-colorable:
Recall that a $3$-coloring of a graph is a function mapping each vertex to one of three colors such that no two adjacent vertices are mapped to the same color. Clearly any subgraph of a graph $G$ admits a $3$-coloring if $G$ does.

Note that in \cref{thm:monotone_refined_intro},
the assumption of $\Phi$ not being $k$-trivial for infinitely many $k$ is
necessary, as otherwise the problem $\#\indsubsprob(\Phi)$ becomes trivial for all $k$
that are greater than a constant only depending on $\Phi$.

Note further, that $\#\W{1}$-completeness in \cref{thm:monotone_refined_intro} is not surprising,
as the decision version of $\#\indsubsprob(\Phi)$ was implicitly shown to be
$\W{1}$-complete by Khot and Raman~\cite{KhotR02}; $\W{1}$ is the decision version of $\#\W{1}$ and should be considered a parameterized equivalent of $\NP$. However, their reduction is not
parsimonious. Also, their proof uses Ramsey's Theorem and thus only yields an implicit
conditional lower bound of $f(k)\cdot |V(G)|^{o\left(\log k\right)}$, whereas our lower
bound is almost tight.

Our second result establishes an almost tight lower bound for sparse properties, that is,
properties $\Phi$ that admit a constant $s$ such that ever graph $H$ for which $\Phi$
holds has at most $s\cdot |V(H)|$ many edges. Furthermore, the bound can be made tight if
the set $\mathcal{K}(\Phi)$ of positive integers $k$ for which $\Phi$ is not $k$-trivial
is additionally \emph{dense}. By this we mean that there is a constant~$\ell$ such
that for every positive integer~$n$, there exists $n \leq k \leq \ell n$ such that $\Phi$
is not $k$-trivial. Note that density rules out artificial properties such as $\Phi(H)=1$
if and only if $H$ is an independent set and has precisely $2\uparrow n$ vertices for some
positive integer~$n$, with $2\uparrow n$ the $n$-fold exponential tower with
base~$2$. In particular, for every property $\Phi$ that is $k$-trivial only for finitely many
$k$, the set $\mathcal{K}(\Phi)$ is dense.
\begin{mtheorem}\label{thm:sparse_intro}
    Let $\Phi$ denote a sparse graph property such that $\Phi$ is not $k$-trivial for
    infinitely many~$k$. Then, $\#\indsubsprob(\Phi)$ is $\#\W{1}$-complete and cannot be solved in time $f(k)\cdot
    |V(G)|^{o\left(k/\log k\right)}$ for any function $f$,
    unless ETH fails.

    If $\mathcal{K}(\Phi)$ is additionally dense, then
    $\#\indsubsprob(\Phi)$ cannot
    be solved in time $f(k)\cdot |V(G)|^{o\left(k\right)}$ for any function $f$,
	unless ETH fails.\lipicsEnd
\end{mtheorem}
Our third result solves the open problem posed by Jerrum and Meeks by proving
that (a strengthened version of) \cref{conj:JM} is true.
\begin{mtheorem}\label{cor:number of edges_refined_intro}
    Let $\Phi$ denote a computable graph property that only depends on the number of edges
    of a graph. If $\Phi$ is not $k$-trivial for infinitely many $k$, then
    $\#\indsubsprob(\Phi)$ is $\#\W{1}$-complete and cannot be solved in time $f(k)\cdot
    |V(G)|^{o\left(k/\log k\right)}$ for any function $f$,
	unless ETH fails.

    If $\mathcal{K}(\Phi)$ is additionally dense, then
     $\#\indsubsprob(\Phi)$ cannot be solved in time $f(k)\cdot
     |V(G)|^{o(k/\sqrt{\log k})}$ for any function $f$,
	unless ETH fails.\lipicsEnd
\end{mtheorem}
Note that, similar to \cref{thm:monotone_refined_intro}, the conditional lower
bounds in the previous two theorems become tight, if ``you cannot beat
treewidth''~\cite{Marx10}; in particular, the condition of being dense can be removed in
that case.

Finally, we consider properties that are \emph{hereditary}, that is, closed under taking
\emph{induced} subgraphs. We obtain a criterion on such graph properties that, if
satisfied, yields a tight conditional lower bound for the complexity of
$\#\indsubsprob(\Phi)$. While the statement of the criterion is deferred to the technical
discussion, we can see that every hereditary property that is defined by a single
forbidden induced subgraph satisfies the criterion.
\begin{mtheorem}\label{thm:hereditary_intro}
    Let $H$ be a graph with at least $2$ vertices and let $\Phi$ denote the property of being
    $H$-free, that is, a graph satisfies $\Phi$ if and only if it does not contain $H$ as an induced subgraph.
    Then,
	$\#\indsubsprob(\Phi)$ is $\#\W{1}$-complete and cannot be solved in time
	$f(k)\cdot |V(G)|^{o(k)}$ for any function $f$, unless ETH fails.\lipicsEnd
\end{mtheorem}
Note that the case of $H$ being the graph with one vertex, which is excluded above, yields the property $\Phi$ which is false on all graphs $G$ with at least one vertex, for which $\#\indsubsprob(\Phi)$ is the constant zero-function and thus trivially solvable.
Hence, \cref{thm:hereditary_intro} establishes indeed a complete classification
for all properties $\Phi$=``$H$-free''.

\subsection{Technical Overview}
We rely on the framework of Complexity Monotonicity of computing linear combinations of
homomorphism counts~\cite{CurticapeanDM17}. More precisely, it is known that for every
computable graph property $\Phi$ and positive integer $k$, there exists a unique
computable function $a$ from graphs to rational numbers such that for all graphs $G$
\begin{equation}\label{eq:gmp_intro}
\#\indsubs{\Phi,k}{G} = \sum_H a(H) \cdot \#\homs{H}{G}\,,
\end{equation}
where $\#\indsubs{\Phi,k}{G}$ denotes the number of induced subgraphs of size $k$ in $G$
that satisfy $\Phi$, and $\#\homs{H}{G}$ denotes the number of graph homomorphisms from
$H$ to $G$. It is known that the function $a$ has finite support, that is, there is only a
finite number of graphs $H$ for which $a(H)\neq 0$.

Intuitively, Complexity Monotonicity states that computing a linear combination of
homomorphism counts is \emph{precisely} as hard as computing its hardest
term~\cite{CurticapeanDM17}. Furthermore, the complexity of computing the number of
homomorphisms from a small graph $H$ to a large graph $G$ is (almost) precisely understood
by the dichotomy result of Dalmau and Jonsson~\cite{DalmauJ04} and the conditional lower
bound under ETH due to Marx~\cite{Marx10}: Roughly speaking, it is possible to compute
$\#\homs{H}{G}$ efficiently if and only if $H$ has small treewidth. As a consequence, the
complexity of computing $\#\indsubs{\Phi,k}{G}$ is precisely determined by the support of the
function $a$. Unfortunately, determining the latter turned out to be an incredibly hard task: It was
shown in~\cite{RothS18} and~\cite{DorflerRSW19} that the function $a$ subsumes a variety
of algebraic and even topological invariants. As a concrete example, a subset of the
authors showed that for edge-monotone properties $\Phi$, the coefficient $a(K_k)$ of the
complete graph in \cref{eq:gmp_intro} is, up to a factor of $k!$, equal to the
reduced Euler characteristic of what is called the simplicial graph complex of $\Phi$ and
$k$~\cite{RothS18}. By this, a connection to Karp's Evasiveness
Conjecture\footnote{Intuitively, the Evasiveness Conjecture states that every decision
tree algorithm verifying a non-trivial edge-monotone graph property has to query every
edge of the input graph in the worst case~\cite{Rosenberg73,Miller13}.} was established.
In particular, it is known that $\Phi$ is evasive on $k$-vertex graphs if the reduced
Euler characteristic is non-zero~\cite{KahnSS84}. As a consequence, the coefficient
$a(K_k)$ can reveal a property to be evasive on $k$-vertex graphs if shown to be non-zero.

The previous example illustrated that identifying the support of the function $a$ in
\cref{eq:gmp_intro} is a hard task, but using the framework of Complexity
Monotonicity requires us to solve this task. In this work, we present a solution for
properties whose $f$-vectors (see below) have low Hamming weight: Given a property
$\Phi$ and a positive integer $k$, we define a $\binom{k}{2}+1$ dimensional vector
$f^{\Phi,k}$ by setting $f^{\Phi,k}_i$ to be the number of edge-subsets of size $i$ of the
complete graph with $k$ vertices such that the induced graph satisfies $\Phi$, that is,
\[f^{\Phi,k}_i := \#\{ A \subseteq E(K_k) ~|~\#A = i \wedge \Phi(K_k[A])=1 \} \]
for all $i=0,\dots,\binom{k}{2}$. By this, we lift the notion of $f$-vectors from abstract
simplicial complexes to graph properties; readers familiar with the latter will observe
that the $f$-vector of an edge-monotone property~$\Phi$ equals the $f$-vector of its
associated graph complex (see for instance \cite{Billera97}). Similarly, we introduce the notions
of $h$-vectors $h^{\Phi,k}$ and $f$-polynomials $\fpol_{\Phi,k}$ of graph properties,
defined as follows; we set $d= \binom{k}{2}$.
\[h^{\Phi,k}_\ell := \sum_{i=0}^\ell (-1)^{\ell-i} \cdot \binom{d - i}{\ell -i}\cdot
    f^{\Phi,k}_i, \text{ where }
    \ell\in\{0,\dots,d\};\quad\text{and}\quad\fpol_{\Phi,k}(x) :=
\sum_{i=0}^d f^{\Phi,k}_i \cdot x^{d-i}\!.\]
Our main combinatorial insight relates the function $a$ in \cref{eq:gmp_intro}
to the $h$-vector of $\Phi$. For the formal statement, we let $\mathcal{H}(\Phi,k,i)$
denote the set of all graphs $H$ with $k$ vertices and $i$ edges that satisfy $\Phi$.
We then show that for all $i =0,\dots, d$, we have
\[ k! \sum_{H \in \mathcal{H}(\Phi,k,i)} a(H) = h^{\Phi,k}_i\!.\]
In particular, the previous equation shows that there is a graph $H$ with $i$ edges
that survives with a non-zero coefficient $a(H)$ in \cref{eq:gmp_intro} whenever the
$i$-th entry of the $h$-vector $h^{\Phi,k}$ is non-zero. As graphs with many edges have
high treewidth, we can thus establish hardness of computing $\#\indsubs{\Phi,k}{G}$ by
proving that there is a non-zero entry with a high index in $h^{\Phi,k}$.
To this end, we relate $h^{\Phi,k}$ and $f^{\Phi,k}$ by observing that their entries
are evaluations of the derivatives of the $f$-polynomial $\fpol_{\Phi,k}(x)$. More
concretely, our goal is to show that a large amount of high-indexed zero entries of
$h^{\Phi,k}$ yields that the only polynomial of degree at most $d$ that satisfies the
constrains given by the evaluations of the derivatives is the zero polynomial. However,
the latter can only be true if $\Phi$ is trivially false on $k$-vertex graphs. Using
Hermite-Birkhoff interpolation and P\'olya's Theorem we are able to achieve this goal
whenever the Hamming weight of $f^{\Phi,k}$ is small. Our meta-theorem thus classifies the complexity of $\#\indsubsprob(\Phi)$ in terms of the Hamming weight of the $f$-vectors of $\Phi$.
\begin{mtheorem}\label{thm:main_general_intro}
    Let $\Phi$ denote a computable graph property and suppose that $\Phi$ is not
    $k$-trivial for infinitely many $k$. Let $\beta: \mathcal{K}(\Phi) \to \mathbb{Z}_{\ge
    0}$ denote the function that maps $k$ to $\binom{k}{2}-\mathsf{hw}(f^{\Phi,k})$.
	If $\beta(k)\in \omega(k)$ then the problem $\#\indsubsprob(\Phi)$ is
	$\#\W{1}$-complete and cannot be solved in time \[g(k)\cdot
    n^{o((\beta(k)/k)/(\log(\beta(k)/k)))}\] for any function $g$, unless
	ETH fails.\lipicsEnd
\end{mtheorem}
For the refined conditional lower bounds in case of monotone properties and properties for which the set $\mathcal{K}(\Phi)$ is dense (see
\cref{thm:monotone_refined_intro,thm:sparse_intro,cor:number of
edges_refined_intro}), we furthermore rely on a consequence of the
Kostochka-Thomason-Theorem~\cite{Kostochka84,Thomason01} that establishes a lower bound on
the size of the smallest clique-minors of graphs with many edges.

In contrast to the previous families of properties, we do not rely on the general
meta-theorem (\cref{thm:main_general_intro}) for our treatment of hereditary properties.
Instead, we carefully construct a
reduction from counting $k$-independent sets in bipartite graphs. Given a hereditary graph
property $\Phi$ defined by the (possibly infinite) set $\Gamma(\Phi)$ of forbidden induced
subgraphs, we say that $\Phi$ is \emph{critical} if there is a graph $H\in \Gamma(\Phi)$
and an edge $e$ of $H$ such that the graph obtained from $H$ by deleting $e$ and then
cloning the former endpoints of $e$ satisfies $\Phi$; the formal definition is provided in
\cref{sec:hereditary}. The reduction from counting $k$-independent sets in
bipartite graphs then yields the following result:
\begin{mtheorem}\label{thm:critical_hardness_intro}
	Let $\Phi$ denote a computable and critical hereditary graph property. Then
	$\#\indsubsprob(\Phi)$ is $\#\W{1}$-complete and cannot be solved in time
	$g(k)\cdot n^{o(k)}$ for any function $g$, unless ETH fails.\lipicsEnd
\end{mtheorem}
We then establish that every hereditary property with precisely one non-trivial forbidden
subgraph $H$ is critical, which yields \cref{thm:hereditary_intro}.

\subsection{Organization of the Paper}
We begin with providing all necessary technical background in \cref{sec:prelims}.
In particular, we introduce the most important notions in parameterized and
fine-grained complexity theory, as well as the principle of Hermite-Birkhoff
interpolation.

\noindent \Cref{sec:main_result} presents and proves our main combinatorial result which
relates the $f$-vectors and $h$-vectors of a computable graph property $\Phi$ on $k$-vertex
graphs to the coefficients in the associated linear combination of homomorphisms as given
by \cref{eq:gmp_intro}.

We derive the meta-theorem for the complexity classification of $\#\indsubsprob(\Phi)$
afterwards in \cref{sec:meta_theorem} and illustrate its applicability by
establishing new conditional lower bounds for properties that are monotone, that have low
edge-densities, and that depend only on the number of edges of a graph. Those bounds are
refined to match the statements of
\cref{thm:monotone_refined_intro,thm:sparse_intro} in \cref{sec:refined_bounds} by
combining our main combinatorial result with results
from extremal graph theory that relate the number of edges of a graph to the size of its
largest clique-minor.

Finally, we present our treatment of hereditary properties in
\cref{sec:hereditary}.

\subsection*{Acknowledgements}

We thank D{\'{a}}niel Marx for pointing out a proof of \cref{lem:hfree_critical} and Jacob Focke for helpful comments.

\section{Preliminaries}\label{sec:prelims}
Given a finite set $S$, we write $\# S$ and $|S|$ for the cardinality of $S$. Further, given a
non-negative integer~$r$, we set $[r]:=\{1,\dots,r\}$; in particular, we have
$[0]=\emptyset$.
The \emph{hamming weight} of a vector $f\in \mathbb{Q}^n$, denoted by
$\mathsf{hw}(f)$, is defined to be the number of non-zero entries of $f$.

\noindent Graphs in this work are simple and do not contain self-loops. Given a graph $G$, we
write $V(G)$ for the vertices and $E(G)$ for the edges of $G$. Furthermore, we define
$\mathcal{G}$ to be the set of all (isomorphism classes of) graphs. The \emph{complement}
$\overline{G}$ of a graph $G$ has vertices $V(G)$ and edges $\overline{E(G)}\setminus \{
\{v,v\}~|~v \in V(G)\}$.

Given a subset $\hat{E}$ of edges of a graph $G$, we write $G[\hat{E}]$ for the graph with
vertices $V(G)$ and edges~$\hat{E}$. Given a subset $\hat{V}$ of vertices of a graph $G$,
we write $G[\hat{V}]$ for the graph with vertices $\hat{V}$ and edges~$E(G)\cap\hat{V}^2$.
In particular, we say that $G[\hat{V}]$ is an \emph{induced subgraph} of $G$. Given graphs
$H$ and $G$, we define $\indsubs{H}{G}$ to be the set of all induced subgraphs of
$G$ that are isomorphic to $H$.

Given graphs $H$ and $G$, a \emph{homomorphism} from $H$ to $G$ is a function $\varphi:
V(H) \rightarrow V(G)$ such that $\{\varphi(u),\varphi(v)\} \in E(G)$ whenever $\{u,v\}
\in E(H)$. We write $\homs{H}{G}$ for the set of all homomorphisms from $H$ to $G$.
In particular, we write $\#\homs{H}{\star}$ for the function that maps a graph $G$ to
$\#\homs{H}{G}$. A bijective homomorphism from a graph $H$ to itself is an
\emph{automorphism} and we write $\auts{H}$ to denote the set of all automorphisms of $H$.

For a graph $H$, we define the {\em average degree} of $H$ as
\[d(H) := \frac{1}{|V(H)|}\cdot \sum_{v \in V(H)} \mathsf{deg}(v).\]
Further, we rely on the \emph{treewidth} of a graph, which is a graph parameter $\mathsf{tw}:
\mathcal{G} \rightarrow \mathbb{N}$. However, we only work with the treewidth in a black-box
manner, and thus we omit the definition and refer the interested reader to the literature,
(see for instance \cite[Chapter 7]{CyganFKLMPPS15}). In particular, we use the following
well-known result from extremal graph theory, which relates the treewidth of a graph $H$
to its average degree.
\begin{lemma}[Folklore, see for instance {\cite[Corollary~1]{ChandranS05}}]\label{lem:convenient}
    For any graph $H$ with average degree at least~$d$, we have $\mathsf{tw}(H)\geq \frac{d}{2}$.\lipicsEnd
\end{lemma}
Finally, we also rely on the following celebrated result from extremal graph theory:
\begin{theorem}[Tur\'an's Theorem, see for instance {\cite[Section 2.1]{Lovasz12}}]
    \label{thm:turan}
    A graph $H$ with more than $\left(1-\frac{1}{r}\right)\cdot \frac{1}{2}|V(H)|^2$ edges
    contains the clique $K_{r+1}$ as a subgraph.\lipicsEnd
\end{theorem}

\paragraph*{Graph Properties}
A \emph{graph property} $\Phi$ is a function from graphs to $\{0,1\}$ such that
$\Phi(H)=\Phi(G)$ whenever $H$ and $G$ are isomorphic. We say that a graph $H$ \emph{satisfies} $\Phi$ if
$\Phi(H)=1$. Given a positive integer $k$ and a graph property~$\Phi$, we write~$\Phi_k$
for the set of all (isomorphism classes of) graphs with $k$ vertices that satisfy $\Phi$.
Furthermore, given a graph $G$, a positive integer $k$, and a graph property $\Phi$, we
write $\indsubs{\Phi,k}{G}$ for the set of all induced subgraphs with $k$ vertices of $G$
that satisfy $\Phi$. In particular, we write $\#\indsubs{\Phi,k}{\star}$ for the function
that maps a graph $G$ to $\#\indsubs{\Phi,k}{G}$.

Given a graph property $\Phi$, we define $\neg\Phi(H) = 1 :\Leftrightarrow \Phi(H) = 0$ as
the \emph{negation} of~$\Phi$.
Furthermore, we define $\overline{\Phi}(H)= 1 :\Leftrightarrow \Phi(\overline{H}) = 1$ as
the \emph{inverse} of~$\Phi$.\footnote{We omit using the word ``complement'' for graph
properties to avoid confusion on whether we mean $\neg\Phi$ or $\overline{\Phi}$.}
We observe the following identities:
\begin{fact}\label{fac:invariance}
    For every graph property $\Phi$, graph $G$ and positive integer $k$, we have
    \begin{align*}
        \#\indsubs{\neg\Phi,k}{G} &= \binom{|V(G)|}{k} - \#\indsubs{\Phi,k}{G}\,, \text{ and}\\
        \#\indsubs{\overline{\Phi},k}{G} &= \#\indsubs{\Phi,k}{\overline{G}} \,.
    \end{align*}
\end{fact}
\begin{proof}
    The first identity is immediate. For the second one, we observe that
    \begin{align*}
        \#\indsubs{\overline{\Phi},k}{G} &= \sum_{H \in \overline{\Phi}_k}
        \#\indsubs{H}{G} = \sum_{\overline{H} \in \Phi_k} \#\indsubs{H}{G}\\
        &= \sum_{\overline{H} \in \Phi_k} \#\indsubs{\overline{H}}{\overline{G}} =
        \#\indsubs{\Phi,k}{\overline{G}},
    \end{align*}
    where we use the equality $\#\indsubs{H}{G} = \#\indsubs{\overline{H}}{\overline{G}}$
    from \cite[Section~5.2.3]{Lovasz12}.
\end{proof}

\paragraph*{Fine-Grained and Parameterized Complexity Theory}
Given a computable graph property $\Phi$, the problem $\#\indsubsprob(\Phi)$ asks, given a
graph $G$ with $n$ vertices and a positive integer $k$, to compute
$\#\indsubs{\Phi,k}{G}$, that is, the number of induced subgraphs of size $k$ in $G$ that
satisfy $\Phi$. Note that the problem can be solved by brute-force in time $f(k)\cdot
O(n^k)$ by iterating over all subsets of $k$ vertices in $G$ and testing which of the
subsets induce a graph that satisfies~$\Phi$; the latter part takes time $f(k)$ for some
$f$ depending on $\Phi$.

As elaborated in the introduction, our goal is to understand the complexity of
$\#\indsubsprob(\Phi)$ for instances with small $k$ and large $n$. More precisely, we wish
to identify the best possible exponent of $n$ in the running time. To this end, we rely on
the frameworks of fine-grained and parameterized complexity theory. Regarding the former, we
prove conditional lower bounds based on the \emph{Exponential Time Hypothesis} due to
Impagliazzo and Paturi~\cite{ImpagliazzoP01}:

\begin{conjecture}[Exponential Time Hypothesis (ETH)]
    The problem $3$-$\textsc{SAT}$ cannot be solved in time $\mathsf{exp}(o(m))$, where $m$ is the
    number of clauses of the input formula.\lipicsEnd
\end{conjecture}
Assuming ETH, we are able to prove that the exponent ($k$) of the brute-force
algorithm for $\#\indsubsprob(\Phi)$ cannot be improved significantly for non-trivial
monotone properties by establishing that no algorithm with a running time of $f(k)\cdot
|V(G)|^{o(k / \sqrt{\log k})}$ for any function $f$ exists.

\noindent In the language of parameterized complexity theory, our reductions also yield
$\#\W{1}$-completeness results, where $\#\W{1}$ should be considered the parameterized
counting equivalent of $\mathrm{NP}$; we provide a rough introduction in what follows and
refer the interested reader to references like \cite{CyganFKLMPPS15} and~\cite{FlumG04} for a detailed treatment.

A \emph{parameterized counting problem} is a pair of a function $P: \Sigma^\ast
\rightarrow \mathbb{N}$ and a computable parameterization $\kappa:\Sigma^\ast \rightarrow
\mathbb{N}$. Examples include the problems $\#\textsc{VertexCover}$ and
$\#\textsc{Clique}$ which ask, given a graph $G$ and a positive integer $k$, to compute
the number $P(G,k)$ of vertex covers or cliques, respectively, of size $k$. Both problems are parameterized by the solution size, that is $\kappa(G,k):=k$. Similarly, the problem
$\#\indsubsprob(\Phi)$ can be viewed as a parameterized counting problem when
parameterized by $\kappa(G,k):=k$; we implicitly assume this parameterization of
$\#\indsubsprob(\Phi)$ in the remainder of this paper.

A parameterized counting problem is called \emph{fixed-parameter tractable} (FPT) if there is a computable function $f$ such that the problem can be solved in time $f(\kappa(x))\cdot |x|^{O(1)}$, where $|x|$ is the input size.

Given two parameterized counting problems $(P,\kappa)$ and $(\hat{P},\hat{\kappa})$, a
\emph{parameterized Turing-reduction} from $(P,\kappa)$ to $(\hat{P},\hat{\kappa})$ is an
algorithm $\mathbb{A}$ that is given oracle access to $\hat{P}$ and, on input $x$,
computes $P(x)$ in time $f(\kappa(x))\cdot |x|^{O(1)}$ for some computable function $f$;
furthermore, the parameter $\kappa(y)$ of every oracle query posed by $\mathbb{A}$ must be
bounded by $g(\kappa(x))$ for some computable function $g$.

While $\#\textsc{VertexCover}$ is known to be fixed-parameter tractable~\cite{FlumG04},
$\#\textsc{Clique}$ is not fixed-parameter tractable, unless ETH
fails~\cite{Chenetal05,Chenetal06}. Moreover,
$\#\textsc{Clique}$ is the canonical complete problem for the parameterized complexity class
$\#\W{1}$, see~\cite{FlumG04}; in particular, we use the following definition of
$\#\W{1}$-completeness in this work.
\begin{definition}
    A parameterized counting problem is $\#\W{1}$-\emph{complete} if it is interreducible
    with $\#\textsc{Clique}$ with respect to parameterized Turing-reductions.
    \lipicsEnd
\end{definition}
Note that that the absence of an FPT algorithm for $\#\textsc{Clique}$ under ETH and the
definition of parameterized Turing-reductions yield that $\#\W{1}$-complete problems are
not fixed-parameter tractable, unless ETH fails, legitimizing the notion of
$\#\W{1}$-completeness as evidence for (fixed-parameter) intractability. Jerrum and Meeks~\cite{JerrumM15} have shown that $\#\indsubsprob(\Phi)$ reduces to $\#\textsc{Clique}$ for every computable property~$\Phi$ with respect to parameterized Turing-reductions. Thus we will only treat the ``hardness part'' of the $\#\W{1}$-completeness results in this paper.

The fine-grained and parameterized complexity of the homomorphism counting problem are
the foundation of the lower bounds established in this work:
Given a class of graphs $\mathcal{H}$, the problem $\#\homsprob(\mathcal{H})$ asks, on input a graph $H \in \mathcal{H}$ and an arbitrary graph $G$, to compute $\#\homs{H}{G}$. The
parameter is given by $|V(H)|$. The following classification shows that, roughly speaking,
the complexity of $\#\homsprob(\mathcal{H})$ is determined by the treewidth of the graphs
in $\mathcal{H}$.

\begin{theorem}[\cite{DalmauJ04,Marx10}]\label{thm:homsdicho}
    Let $\mathcal{H}$ denote a recursively enumerable class of graphs. If the treewidth of~$\mathcal{H}$ is bounded by a constant, then $\#\homsprob(\mathcal{H})$ is solvable in polynomial time. Otherwise, the problem is $\#\W{1}$-complete and cannot be solved in time
    \[f(|V(H)|)\cdot |V(G)|^{o\left(\frac{\mathsf{tw}(H)}{\log \mathsf{tw}(H)}\right)} \]
    for any function $f$, unless ETH fails.
\end{theorem}
Note that the classification of $\#\homsprob(\mathcal{H})$ into polynomial-time
and $\#\W{1}$-complete cases is explicitly stated and proved in the work of Dalmau and
Jonson~\cite{DalmauJ04}. However, the conditional lower bound follows only implicitly by a
result of Marx~\cite{Marx10}. We provide a proof for completeness.
\pagebreak
\begin{proof}
    As discussed, we focus on the lower bound, which
    follows implicitly from a result of Marx~\cite{Marx10} on
    the complexity of \emph{Partitioned Subgraph Isomorphism}: The problem
    $\textsc{PartitionedSub}(\mathcal{H})$ asks, given a graph $H \in \mathcal{H}$, an
    arbitrary graph $G$ and a (not necessarily proper) vertex coloring $c:V(G)\rightarrow
    V(H)$, to \emph{decide} whether there is an injective homomorphism $\varphi$ from $H$
    to $G$ such that $c(\varphi(v)) = v$ for each vertex~$v$ of $H$.

    The result of Marx~\cite[Corollary~6.2]{Marx10} states that for every $\mathcal{H}$ of
    unbounded treewidth, the problem $\textsc{PartitionedSub}(\mathcal{H})$ cannot be
    solved in time
    \begin{equation}\label{eq:marx_bound}
        f(|V(H)|)\cdot |V(G)|^{o({\mathsf{tw}(H)}/{\log \mathsf{tw}(H)})}
    \end{equation}
    for any function $f$, unless ETH fails.
    Now suppose we are given a graph $H \in \mathcal{H}$, an arbitrary graph $G$ and a
    coloring $c:V(G)\rightarrow V(H)$. We wish to decide whether there is an
    injective homomorphism~$\varphi$ from~$H$ to~$G$ such that $c(\varphi(v)) = v$ holds for
    each vertex $v$ of~$H$. Note first, that we can drop the requirement of
    $\varphi$ being injective, as every homomorphism that preserves the coloring is
    injective. Note further, that without loss of generality,
    we can assume that~$c$ is a homomorphism from
    $G$ to $H$: Every edge $\{u,v\}$ of~$G$ such that $\{c(u),c(v)\} \notin E(H)$ is
    irrelevant for finding a homomorphism $\varphi$ from $H$ to $G$ that preserves the
    coloring~$c$. Hence, we can delete all of those edges from $G$. Thus, the
    problem $\textsc{PartitionedSub}(\mathcal{H})$ is equivalent to the problem
    $\textsc{cp}\text{-}\homsprob(\mathcal{H})$ which asks, given a graph $H \in
    \mathcal{H}$, an arbitrary graph~$G$, and a homomorphism $c \in \homs{G}{H}$, to decide
    whether there is a $\varphi \in \homs{H}{G}$ such that $c(\varphi(v))
    = v$ for each~$v\in V(H)$. Finally, it is known that (the counting version of)
    $\textsc{cp}\text{-}\homsprob(\mathcal{H})$ tightly reduces to
    $\#\homsprob(\mathcal{H})$ via the principle of inclusion and exclusion
    \cite[Lemma~2.52]{Roth19} or polynomial
    interpolation \cite[Section~3.2]{DellRW19icalp}.
    Thus the conditional lower bound
    in~\cref{eq:marx_bound} holds for $\#\homsprob(\mathcal{H})$ as well.
\end{proof}
The question whether the lower bound from \cref{thm:homsdicho}
can be strengthened to $f(|V(H)|)\cdot
|V(G)|^{o(\mathsf{tw}(H))}$ is known as ``Can you beat treewidth?'' and constitutes a major open problem in parameterized complexity theory and
an obstruction for tight conditional lower bounds on the complexity of a variety of
(parameterized) problems, see for instance
\cite{LokshtanovMS11,Curticapean15,CurticapeanM14,CurticapeanDM17}.

As described in the introduction, the complexity of computing a finite linear combination
of homomorphism counts is precisely determined by the complexity of computing the
non-vanishing terms. The formal statement is provided subsequently.
\begin{theorem}[Complexity Monotonicity~\cite{ChenM16,CurticapeanDM17}]\label{thm:monotonicity}
    Let $a:\mathcal{G} \rightarrow \mathbb{Q}$ denote a function of finite support and let $F$
    denote a graph such that $a(F) \neq 0$. There are a computable function $g$ and a
    deterministic algorithm $\mathbb{A}$ with oracle access to the
    function \[G \mapsto \sum_{H \in \mathcal{G}} a(H)\cdot \#\homs{H}{G},\]
    and which, given a graph $G$ with $n$ vertices, computes $\#\homs{F}{G}$ in time
    $g(a)\cdot n^c$, where $c$ is a constant independent of $a$. Furthermore, each queried
    graph has at most $g(a)\cdot n$ vertices.\lipicsEnd
\end{theorem}

As observed by Curticapean, Dell and Marx~\cite{CurticapeanDM17}, counting induced
subgraphs of size $k$ that satisfy~$\Phi$ is equivalent to computing a finite linear
combination of homomorphism counts. Thus, the previous results yield an \emph{implicit}
dichotomy for $\#\indsubsprob(\Phi)$.

\begin{theorem}[\cite{CurticapeanDM17}]\label{thm:impl_dicho}
    Let $\Phi$ denote a computable graph property and let $k$ denote a positive integer.
    There is a unique and computable function $a:\mathcal{G} \rightarrow \mathbb{Q}$
    of finite support such that
    \[
        \#\indsubs{\Phi,k}{\star} = \sum_{H \in \mathcal{G}} a(H) \cdot \#\homs{H}{\star}.
    \]
    Furthermore, the problem $\#\indsubsprob(\Phi)$ is either fixed-parameter tractable or
    $\#\W{1}$-complete.\lipicsEnd
\end{theorem}
Note that the result on $\#\indsubsprob(\Phi)$ in the previous theorem does not concern
the fine-grained complexity of the problem. To reveal the latter, it is necessary to
understand the support of the function $a$; we tackle this task in detail in
\cref{sec:main_result}.

\paragraph*{$\bm{f}$-Vectors and $\bm{h}$-Vectors}
It was observed in~\cite{RothS18} that there is a close connection between the structure
of the simplicial graph complex of edge-monotone properties $\Phi$ and the complexity of
$\#\indsubsprob(\Phi)$. In this work, we generalize two important topological
invariants of simplicial complexes to arbitrary graph properties: The $f$-vector and the
$h$-vector.

\begin{definition}\label{def:fhvectors}
    Let $\Phi$ denote a graph property, let $k$ denote a positive integer and set $d=\binom{k}{2}$.
    The $f$\emph{-vector} $f^{\Phi,k} = (f^{\Phi,k}_i)_{i=0}^d$ of $\Phi$ and $k$ is defined by
    \[f^{\Phi,k}_i := \#\{ A \subseteq E(K_k) ~|~\#A = i \wedge \Phi(K_k[A])=1 \}\,,
    \text{ where } i\in\{0,\dots,d\}, \]
    that is, $f^{\Phi,k}_i$ is the number of edge-subsets of size $i$ of $K_k$ such that
    the induced graph satisfies~$\Phi$.

    The $h$\emph{-vector} $h^{\Phi,k} = (h^{\Phi,k}_\ell)_{\ell=0}^d$  is defined by
    \[h^{\Phi,k}_\ell := \sum_{i=0}^\ell (-1)^{\ell-i} \cdot \binom{d - i}{\ell -i}\cdot
    f^{\Phi,k}_i\,, \text{ where } \ell\in\{0,\dots,d\}.\lipicsEnd\]
\end{definition}
As mentioned before, note that those notions of $f$ and $h$-vectors correspond
to the eponymous notions for simplicial (graph) complexes.\footnote{In some parts of the
literature, the $f$-vector comes with an index shift of $-1$ due to the topological
interpretation of simplicial complexes.} We omit the definition of the latter as we are
only concerned with the generalized notions and refer the interested reader e.g. to \cite{Billera97}.


It turns out that the \emph{non-vanishing} of suitable entries $h^{\Phi,k}_\ell$ of the
$h$-vector implies hardness for $\#\indsubsprob(\Phi)$. The result in~\cite{RothS18} can
be considered as a very restricted special case as it shows that the non-vanishing of the reduced
Euler characteristic of the complex (which is equal to the entry $h^{\Phi,k}_d$) implies
hardness. On the other hand, for many graph properties it is easy to deduce information
about the $f$-vector (for instance that $f^{\Phi,k}_\ell=0$ for sufficiently large $\ell$ with
respect to $k$). We observe that the $f$ and $h$-vectors of a graph property are
related by the so-called the $f$-polynomial which is again a
generalization of the epynomous notion for simplicial complexes:

\begin{definition}
    Let $\Phi$ denote a graph property, let $k$ denote a positive integer and set $d=\binom{k}{2}$.
    The $f$\emph{-polynomial} of~$\Phi$ and~$k$ is a univariate polynomial of degree at most $d$
    defined as follows:
    \[ \fpol_{\Phi,k}(x) := \sum_{i=0}^d f^{\Phi,k}_i \cdot x^{d-i}\!.\lipicsEnd\]
\end{definition}
As we see in the proof of  Lemma \ref{lem:birkhoff_interpol}, the entries of the $f$ and $h$-vectors are given up to combinatorial factors by derivatives of the $f$-polynomial at $0$ and $-1$. Intuitively, we apply Hermite-Birkhoff interpolation on $\fpol_{\Phi,k}$ and
its derivatives to prove that specific entries of $h^{\Phi,k}$ cannot vanish in case a
sufficient number of entries of~$f^{\Phi,k}$ do, unless $\Phi$ is trivially false on
$k$-vertex graphs.

\paragraph*{Hermite-Birkhoff Interpolation and P\'olya's Theorem}
While a univariate polynomial of degree $d$ is uniquely determined by $d+1$ evaluations in
pairwise different points, the problem of \emph{Hermite-Birkhoff interpolation} asks under
which conditions we can uniquely recover the polynomial if we instead impose conditions on
the derivatives of the polynomial at $m$ distinct points.
Following the notation of Schoenberg~\cite{Schoenberg66}, the problem is formally
expressed as follows.
Given a matrix $E=(\varepsilon_{ij})\in\{0,1\}^{m \times d+1}$ where $i \in \{1,\dots,m\}$
and $j \in \{0,\dots d\}$, as well as reals~$x_1 < \dots < x_m$, the goal is to find a
polynomial~$\fpol$ of degree at most $d$ such that for all $i$ and $j$ with $\varepsilon_{ij}
= 1$ we have \[\fpol^{(j)}(x_i) = 0\]
Here, $\fpol^{(j)}$ denotes the $j$-th derivative of $\fpol$. In particular, we are
interested under which conditions on the matrix $E$, the zero polynomial is the
\emph{unique} solution. In this case, $E$ is called \emph{poised}. It turns out that
the case $m=2$ is sufficient for our purposes; fortunately, this case was fully solved by
P\'olya:

\begin{theorem}[P\'olya's Theorem~\cite{Polya31,Schoenberg66}]\label{thm:polya}
    Let $E$ be defined as above with $m=2$. Suppose that $\sum_{i,j}\varepsilon_{ij} = d+1$
    and for every $j \in \{0,\dots,d\}$ set
    \[M_j := \sum_{i=0}^j \varepsilon_{1,i} + \varepsilon_{2,i}. \]
    Then, $E$ is poised if and only if $M_j \geq j+1$ holds true for all $j \in
    \{0,\dots,d-1\}$.\lipicsEnd
\end{theorem}

\section{Homomorphism Vectors of Graph Properties}\label{sec:main_result}

In this section we discuss and prove our main technical result:
\begin{restatable}{theorem}{mnthmcom}\label{thm:main_theorem_combinatorial}
    Let $\Phi$ denote a computable graph property, let $k$ denote a positive integer,
    and let~$w$ denote the Hamming weight of the $f$-vector $f^{\Phi,k}$.
    Suppose that $\Phi$ is not trivially false on $k$-vertex graphs.
    Then there is a unique and computable function
    $a:\mathcal{G} \rightarrow \mathbb{Q}$ of finite support such that
    \[\#\indsubs{\Phi,k}{\star} = \sum_{H \in \mathcal{G}} a(H) \cdot \#\homs{H}{\star},\]
    satisfying that there is a graph $K$ on $k$ vertices and at least $\binom{k}{2}-w+1$
edges such that $a(K) \neq 0$.\ifx\mthmrst\undefined{\lipicsEnd}\fi
\end{restatable}
\def\mthmrst{1}
\noindent First, recall from \cref{thm:impl_dicho} that for any computable graph property $\Phi$ and
positive integer $k$,
there is a unique computable function $a:\mathcal{G} \rightarrow \mathbb{Q}$
(with finite support) satisfying
\begin{equation}\label{eq:indsubgmp}
\#\indsubs{\Phi,k}{\star} = \sum_{H \in \mathcal{G}} a(H) \cdot \#\homs{H}{\star}.
\end{equation}
Now, for the remainder of the section, fix a (computable) graph property $\Phi$ and
a positive integer $k$ (and thus the function $a$).
This allows us to simplify the notation for the $f$ and $h$-vectors, as well as
for the $f$-polynomial: We write $f := f^{\Phi, k}\!$, $h := h^{\Phi, k}\!$, and $\fpol :=
\fpol_{\Phi, k}$.
Furthermore, we set $d := \binom{k}{2}$ and we write~$\mathcal{H}_i$ for the set of all graphs
on $k$ vertices and with $i$ edges.

Next, we define the vector $\tilde{h}_i$ as \[
    \tilde{h}_i := \sum_{K \in \mathcal{H}_i} a(K) \,, \text{ where } i\in\{0,\dots,d\},
\] that is, the $i$-th entry of $\tilde{h}$ is the sum of the coefficients of graphs with
$k$ vertices and $i$ edges in~\cref{eq:indsubgmp}.
Now we  establish the aforementioned connection between the coefficients
of~\cref{eq:indsubgmp} and the $h$-vector of the property~$\Phi$.
\begin{lemma}\label{lem:coef_sums}
    We have $k! \cdot \tilde{h} = h$.
\end{lemma}
Note that as a consequence, the $h$-vector of a simplicial graph complex
is determined by the coefficients of its associated linear combination of homomorphisms.
\begin{proof}
    Given two graphs $H$ and $H'$ on $k$ vertices each, we write $\#\{H'\supseteq H\}$ for the
    number of possibilities of adding edges to $H$ such that (a graph isomorphic to) $H'$
    is obtained. We start with the following claim which was implicitly shown
    in~\cite{RothS18}; we include a proof for completeness.
    \begin{claim}\label{clm:single_coef}
        Let $K$ denote a graph with $k$ vertices and define $a$ as in~\cref{eq:indsubgmp}.
        We have
        \[a(K) = \sum_{H \in \Phi_k}\#\auts{H}^{-1}\cdot (-1)^{\#E(K)-\#E(H)} \cdot \#\{K
        \supseteq H\}. \]
    \end{claim}
    \begin{claimproof}
        Fix a graph $K$ with $k$ vertices.
        Using the standard transformations from strong embeddings to embeddings and from
        embeddings to homomorphisms (see for instance \lovasz~\cite{Lovasz12}), we obtain
        the following:\footnote{This step is done explicitly in~\cite{RothS18}.}
        \begin{align*}
            &\#\indsubs{\Phi,k}{\star}\\
            &\quad=\sum_{H\in \Phi_k} \#\auts{H}^{-1} \sum_{H' \in
            \mathcal{G}} (-1)^{\#E(H')-\#E(H)} \cdot  \#\{H'\supseteq H\}  \sum_{\rho \geq
            \emptyset} \mu(\emptyset,\rho) \cdot  \#\homs{H'/\rho}{\star},
        \end{align*}
        where $\mu$ denotes the Möbius function and the rightmost sum ranges over the
        partition lattice of the set of vertices of $H'$. Furthermore, $H'/\rho$ is the
        quotient graph obtained by identifying vertices of $H'$ along the partition~$\rho$.
        In particular, $H'/\emptyset = H'$. We omit the details, which can be
        found in~\cite{Lovasz12,RothS18}, as we only need that $H'/\rho$ has strictly less
        than $k$ vertices for all $\rho > \emptyset$ and that
        $\mu(\emptyset,\emptyset)=1$. This allows us to rewrite the previous equation as
        follows:
        \begin{align*}
            &\#\indsubs{\Phi,k}{\star} \\
            &\quad=\sum_{H\in \Phi_k} \#\auts{H}^{-1} \sum_{H' \in \mathcal{G}}
            (-1)^{\#E(H')-\#E(H)}  \#\{H'\supseteq H\} \cdot \#\homs{H'}{\star} +
            R(\Phi_k),
        \end{align*}
        where the remainder $R(\Phi_k)$ does not depend on any numbers $\#\homs{F}{\star}$ for
        graphs $F$ with $k$ vertices. In particular, reordering and grouping the
        coefficients of $\#\homs{K}{\star}$ yields the claim.
    \end{claimproof}
    Next, we investigate the term $\sum_{K \in \mathcal{H}_\ell}\#\{K \supseteq H\}$.
    \begin{claim}\label{clm:collect}
        Let $\ell\in\{0,\dots,d\}$ denote an integer and let $H$ denote a graph with $k$ vertices
        and at most $\ell$ edges. Then, we have \[
            \sum_{K \in \mathcal{H}_\ell}\!\!\!
            \#\{K \supseteq H\} =
            \binom{d - \#E(H)}{\ell - \#E(H)}.
        \]
    \end{claim}
    \begin{claimproof}
        Any extension from the graph $H$ to a graph with $\ell$ edges has to add
        $\ell - \#E(H)$ edges to $H$; there are exactly $d - \#E(H)$
        possible choices for these $\ell - \#E(H)$ edges. Hence the claim follows
        from basic combinatorics.
    \end{claimproof}
    Now, fix an $\ell \in \{0,\dots,d\}$; we proceed to show that
    $k!\cdot \tilde{h}_\ell = h_\ell$, which proves the lemma. To that end, from the
    definition of $\tilde{h}$, we obtain
    \begin{align*}
        k! \cdot \tilde{h}_\ell &= k!\cdot \sum_{K \in \mathcal{H}_\ell} a(K)\\
        ~&= k!\cdot \sum_{K \in \mathcal{H}_\ell} \sum_{H \in \Phi_k}\#\auts{H}^{-1}\cdot
        (-1)^{\ell-\#E(H)} \cdot \#\{K \supseteq H\}\\
        ~&=\sum_{H \in \Phi_k} k!\cdot \#\auts{H}^{-1}\cdot (-1)^{\ell-\#E(H)} \sum_{K \in
        \mathcal{H}_\ell} \#\{K \supseteq H\}\,,
    \end{align*}
    where the second equality holds due to \cref{clm:single_coef}.
    Now observe that $\#\{K \supseteq H\} = 0$ if $H$ has more edges than $K$. Thus we see that
    \begin{align*}
        k! \cdot \tilde{h}_\ell &= \sum_{\substack{H \in \Phi_k\\ \#E(H)\leq \ell}}
        k!\cdot \#\auts{H}^{-1}\cdot (-1)^{\ell-\#E(H)} \sum_{K \in \mathcal{H}_\ell}
        \#\{K \supseteq H\}\\
        ~&= \sum_{\substack{H \in \Phi_k\\ \#E(H)\leq \ell}} k!\cdot \#\auts{H}^{-1}\cdot
        (-1)^{\ell-\#E(H)} \cdot  \binom{d - \#E(H)}{\ell - \#E(H)}\,,
    \end{align*}
    where the last equality holds due to \cref{clm:collect}. Next we use the fact that $k!$ is
    the order of the symmetric group $\mathsf{Sym}_k$:
    For any graph $H$ in the above sum, choose a set $A$ of edges of the labeled complete
    graph~$K_k$ on $k$ vertices such that the corresponding subgraph $K_k[A]$ is
    isomorphic to $H$. The group $\mathsf{Sym}_k$ acts on the vertices and thus on the
    edges of $K_k$. By the definition of a graph automorphism, the stabilizer of the
    set $A$ has exactly $\#\auts{H}$ elements.

    \noindent Now observe that the orbit of $A$ under $\mathsf{Sym}_k$ is the collection of all sets
    $A'$ such that $K_k[A']\cong H$. Therefore, by the Orbit Stabilizer Theorem, we have
    \[k! \cdot \#\auts{H}^{-1} = \#\{A' \subseteq E(K_k)~|~K_k[A'] \cong H\}.\]
    Hence we can conclude that
    \begin{align*}
        k! \cdot \tilde{h}_\ell &= \sum_{\substack{H \in \Phi_k\\ \#E(H)\leq \ell}} \#\{A
        \subseteq E(K_k)~|~K_k[A] \cong H\}\cdot (-1)^{\ell-\#E(H)} \cdot  \binom{d -
    \#E(H)}{\ell - \#E(H)}\\
        ~&=\sum_{i=0}^\ell \sum_{\substack{H \in \Phi_k\\ \#E(H)= i}} \#\{A \subseteq
        E(K_k)~|~K_k[A] \cong H\}\cdot (-1)^{\ell-i} \cdot  \binom{d - i}{\ell - i}\\
        ~&=\sum_{i=0}^\ell \#\{ A \subseteq E(K_k) ~|~\#A = i \wedge \Phi(K_k[A])=1
        \}\cdot (-1)^{\ell-i} \cdot  \binom{d - i}{\ell - i}\\
        ~&= \sum_{i=0}^\ell f_i\cdot (-1)^{\ell-i} \cdot  \binom{d - i}{\ell - i} =
        h_\ell,
    \end{align*}
    completing the proof.
\end{proof}
In the next step, we use P\'olya's Theorem to prove that the Hamming weight of the
$f$-vector determines an index $\beta$ of the $h$-vector such that at least one entry of
$h$ with index at least~$\beta$ is non-zero. By \cref{lem:coef_sums} the same
then follows for $\tilde{h}$.

\begin{lemma}\label{lem:birkhoff_interpol}
    Let $w$ denote the Hamming weight of $f$ and set $\beta = d-w$. If $\Phi$ is not trivially
    false on $k$-vertex graphs then at least one of the values $h_d,\dots,h_{\beta+1}$ is non-zero.
\end{lemma}
\begin{proof}
    Recall the definition of the $f$-polynomial $\fpol(x)=\sum_{i=0}^d f_i \cdot x^{d-i}$
    and observe that \[
        \fpol^{(j)}(x) = \sum_{i=0}^{d-j} f_i \cdot (d-i)_j \cdot x^{d-j-i}.
    \] By $j_j = j!$, we immediately obtain $\fpol^{(j)}(0) = f_{d-j} \cdot j!$.
    Therefore, by assumption, we have $\fpol^{(j)}(0) = 0$ for $\beta+1$ many indices $j$.

Furthermore, we see that
    \begin{alignat*}{3}
        \fpol^{(j)}(-1) &= \sum_{i=0}^{d-j} f_i \cdot (d-i)_j \cdot (-1)^{d-j-i}
                        &&= j! \cdot \sum_{i=0}^{d-j} f_i \cdot \binom{d-i}{j} \cdot (-1)^{d-j-i}\\
        ~&= j! \cdot \sum_{i=0}^{d-j} f_i \cdot \binom{d-i}{(d-j)-i} \cdot (-1)^{d-j-i}~
         &&= j! \cdot h_{d-j}.
    \end{alignat*}

    Now assume for the sake of contradiction that each of the values $h_d,\dots,h_{\beta+1}$ is
    zero. Consequently, $\fpol^{(j)}(-1)=0$ for $j=0,\dots, w-1$.  Interpreting those
    evaluations of the derivatives of the $f$-polynomial as an instance of
    Hermite-Birkhoff interpolation, the corresponding matrix $E$ looks as
    follows:\footnote{Recall that an entry~$1$ in the matrix $E$ represents an evaluation
    $\fpol^{(j)}(-1)=0$ in the first row and an evaluation $\fpol^{(j)}(0)=0$ in the
    second row.}
    \[
        \begin{blockarray}{rcccccccc}
            ~ & 0 & 1 & 2 & \dots & w-1 & w& \dots & d \\[1em]
            \begin{block}{l(cccccccc)}
                ~& ~1 & 1 & 1 & \dots & 1 & 0 & \dots & 0 \\
                ~ & \varepsilon_{20} & \varepsilon_{21} & \varepsilon_{22} & \dots &
                \varepsilon_{2(w-1)}& \varepsilon_{2w} & \dots & \varepsilon_{2d}~ \\
            \end{block}
        \end{blockarray}
    \]
    In particular, at least $\beta+1=d+1-w$ of the values $\varepsilon_{2j}$ are $1$;
    As $\beta+1$ and $w$ sum up to $d+1$, we can easily verify that the conditions of
    P\'olya's Theorem (\cref{thm:polya}) are satisfied:
    Let us modify $E$ by arbitrarily choosing \emph{precisely} $\beta+1$ of the
    $\varepsilon_{2,j}$ that are $1$ and set the others to $0$, and call the resulting
    matrix $\hat{E}$.  We then have both $M_j \geq j+1$ (for all $j \in
    \{0,\dots,d-1\}$) and the first and second row of $\hat{E}$ sum up to
    \emph{precisely} $d+1$. Hence the matrix $\hat{E}$ is poised, that is, the only
    polynomial of degree at most $d$
    that satisfies the corresponding instance of Hermite-Birkhoff interpolation is the zero
    polynomial. As we obtained $\hat{E}$ from $E$ just by ignoring some vanishing
    conditions, the same conclusion is true for $E$ and thus $\fpol=0$ is the unique
    solution. This, however, contradicts the fact that the property $\Phi$ is not
    trivially false on $k$-vertex graphs, completing the proof.
\end{proof}

Combining \cref{lem:coef_sums,lem:birkhoff_interpol} yields our main technical result,
which we restate here for convenience.
\mnthmcom*
\begin{proof}
    Set $d=\binom{k}{2}$ and $\beta=d-w$. By \cref{eq:indsubgmp} the function $a$ exists
    and is computable and has a finite support.
    Now, \cref{lem:birkhoff_interpol} implies that at least
    one of the values $h^{\Phi,k}_d,\dots,h^{\Phi,k}_{\beta+1}$ is non-zero
    and thus, by \cref{lem:coef_sums}, at least one of the values
    $\tilde{h}_d,\dots,\tilde{h}_{\beta+1}$ is non-zero as well.
    Next, observe that $\tilde{h}_i = \sum_{K \in \mathcal{H}_i} a(K)$ for all
    $i \in \{0,\dots,d\}$, where $\mathcal{H}_i$ is the set of all graphs on $k$ vertices
    and $i$ edges. In particular, $\tilde{h}_i \neq 0$ implies that $a(K)\neq 0$ for at least
    one $K \in \mathcal{H}_i$, yielding the claim.
\end{proof}

\section{A Classification of
    \texorpdfstring{\#IndSub$\bm{(\Phi)}$}{{\#IndSub(Phi)}} by the Hamming
Weight of the \texorpdfstring{$\bm{f}$}{f}-Vectors}\label{sec:meta_theorem}
In this section, we derive a general hardness result for $\#\indsubsprob(\Phi)$ based on the
Hamming weight of the $f$-vector.
In a sense, we ``black-box'' \cref{thm:main_theorem_combinatorial}; using the resulting
classification, we establish first hardness results and almost tight conditional lower bounds for a
variety of families of graph properties.

However, note that taking a closer look at the number of edges of the graphs with
non-vanishing coefficients (as provided by \cref{thm:main_theorem_combinatorial})
often yields improved, sometimes even matching conditional lower bounds; we defer the
treatment of the refined analysis to \cref{sec:refined_bounds}.

In what follows, we write $\mathcal{K}(\Phi)$ for the set of all $k$ such that $\Phi_k$ is non-empty.

\begin{theorem}\label{thm:main_general}
    Let $\Phi$ denote a computable graph property and suppose that the set $\mathcal{K}(\Phi)$ is
    infinite. Let $\beta: \mathcal{K}(\Phi) \to \mathbb{Z}_{\ge 0}$ denote the
    function that maps $k$ to $\binom{k}{2}-\mathsf{hw}(f^{\Phi,k})$.
    If $\beta(k)\in \omega(k)$ then the problem $\#\indsubsprob(\Phi)$ is
    $\#\W{1}$-complete and cannot be solved in time \[g(k)\cdot
    |V(G)|^{o({(\beta(k)/k)}/{\log(\beta(k)/k)})}\] for any function $g$, unless
    ETH fails. The same statement holds for $\#\indsubsprob(\overline{\Phi})$ and
    $\#\indsubsprob(\neg\Phi)$.
\end{theorem}
Note that the condition of $\mathcal{K}(\Phi)$ being infinite is necessary for
hardness: Otherwise there is a constant~$c$ such that we can output $0$ whenever
$k\geq c$ and solve the problem by brute-force if $k< c$, yielding an algorithm with a
polynomial running time.
Note further that the $(\log(\beta(k)/k))^{-1}$-factor in the exponent
is related to the question of whether it is possible to
``beat treewidth''~\cite{Marx10}.\footnote{See \cref{thm:homsdicho} and its discussion.}
In particular, if the factor of $(\log \mathsf{tw}(H))^{-1}$ in \cref{thm:homsdicho} can
be dropped, then all further results in this section can be strengthened to yield tight conditional
lower bounds under ETH.
\begin{proof}
    By \cref{thm:main_theorem_combinatorial}, for each $k\in \mathcal{K}(\Phi)$ we obtain
    a graph $H_k$ with $k$ vertices and at
    least $\beta(k)$ edges such that $a(H_k)\neq 0$, where $a$ is the function
    in~\cref{eq:indsubgmp}. The average degree of $H_k$ satisfies
    \[d(H_k) = \frac{1}{k}\cdot \sum_{v \in V(H_k)}\mathsf{deg}(v) = \frac{2|E(H_k)|}{k}
    \geq \frac{2\beta(k)}{k} \,,\]
    where the second equality is due to the Handshaking Lemma. By
    \cref{lem:convenient}, we thus obtain that $\mathsf{tw}(H_k) \geq
    \frac{\beta(k)}{k}$, which is unbounded as $\beta(k)\in \omega(k)$ by assumption.

    Now let $\mathcal{H}$ denote the set of all graphs $H_k$ for $k \in \mathcal{K}(\Phi)$.
    By \cref{thm:homsdicho}, we obtain that $\#\homsprob(\mathcal{H})$ is
    $\#\W{1}$-complete and cannot be solved in time
    \[g(k)\cdot |V(G)|^{o(({\beta(k)/k})/{\log(\beta(k)/k)})}\] for any function
    $g$, unless ETH fails. Further, by Complexity Monotonicity
    (\cref{thm:monotonicity}), the same is true for $\#\indsubsprob(\Phi)$ as well.
    Finally, we use \cref{fac:invariance} to obtain the same result for
    $\#\indsubsprob(\overline{\Phi})$ and $\#\indsubsprob(\neg\Phi)$; completing the
    proof.
\end{proof}

\subsection{Low Edge-Densities and Sparse Graph Properties}\label{sec:low_edge_dens}

As a first application of \cref{thm:main_general},
we consider properties $\Phi$ that satisfy \[\mathsf{hw}(f^{\Phi,k}) \in o(k^2).\]
We say that such a property $\Phi$ has a \emph{low edge-densities}.
Properties with low edge-density subsume, for example,
exclusion of a set of fixed minors such as planarity. They have been studied by
Jerrum and Meeks~\cite{JerrumM15density}, where they show that $\#\indsubsprob(\Phi)$
is $\#\W{1}$-complete for these properties.
However, their proof uses Ramsey's Theorem and thus only establishes an implicit conditional
lower bound of $g(k) \cdot |V(G)|^{o(\log k)}$. In contrast, we achieve the following, almost
tight lower bound:

\begin{theorem}\label{cor:low_edge_densities}
    Let $\Phi$ denote a computable graph property with low edge-densities.
    Suppose that the set~$\mathcal{K}(\Phi)$ is infinite.
    Then $\#\indsubsprob(\Phi)$ is $\#\W{1}$-complete
    and cannot be solved in time \[g(k)\cdot |V(G)|^{o\left(k/\log k\right)}\] for any function
    $g$, unless ETH fails. The same is true for $\#\indsubsprob(\overline{\Phi})$ and
    $\#\indsubsprob(\neg\Phi)$.
\end{theorem}
\begin{proof}
    If $\Phi$ has low edge-densities, then we have $\beta(k)=\binom{k}{2}- \mathsf{hw}(f^{\Phi,k})
    \in \Theta(k^2)$. Thus \[
    o\left(\frac{\beta(k)/k}{\log(\beta(k)/k)}\right) = o\left(k/\log k\right) \,.\]
    The claim hence follows by \cref{thm:main_general}.
\end{proof}
The previous result applies, in particular, to sparse properties. In \cref{sec:refined_bounds} we show, that a refined analysis based on
Tur\'an's Theorem as well as \cref{thm:main_theorem_combinatorial} establishes a tight conditional lower bound for sparse properties $\Phi$ that additionally satisfy a density condition on $\mathcal{K}(\Phi)$; the combination of those two results then implies \cref{thm:sparse_intro}.

\subsection{Graph Properties Depending Only on the Number of Edges}
Jerrum and Meeks~\cite{JerrumM15density,JerrumM17} asked whether
$\#\indsubsprob(\Phi)$ is $\#\W{1}$-complete whenever $\Phi$ is non-trivial infinitely
often and only depends on the number of edges of a graph, that is,
\[\forall H_1, H_2: |E(H_1)|=|E(H_2)| \Rightarrow \Phi(H_1) = \Phi(H_2).\]
We answer this question affirmatively, even for properties that can depend both on the number of edges and vertices of the graph, and additionally provide an almost tight conditional
lower bound:

\begin{theorem}\label{cor:number of edges}
    Let $\Phi$ denote a computable graph property that only depends on the number of edges and the number of vertices of a
    graph. If $\Phi_k$ is non-trivial only for finitely many $k$ then $\#\indsubsprob(\Phi)$
    is fixed-parameter tractable. Otherwise, $\#\indsubsprob(\Phi)$ is $\#\W{1}$-complete and
    cannot be solved in time \[g(k)\cdot |V(G)|^{o\left(k/\log k\right)}\] for any function $g$,
    unless ETH fails.
\end{theorem}
Note that \cref{cor:number of edges} is also true for $\#\indsubsprob(\overline{\Phi})$
and $\#\indsubsprob(\neg\Phi)$, as $\neg\Phi$ and $\overline{\Phi}$ depend only on the
number of edges and vertices of a graph if and only if $\Phi$ does.
\begin{proof}
    First, assume that $\Phi_k$ is non-trivial only for finitely many $k$. Then, there
    is a constant $c$ such that for every $k > c$, the property $\Phi_k$ is either
    trivially true or trivially false.
    Hence, given as input a graph $G$ and an integer $k$, we check whether $k\leq c$.
    If this is the case, we solve the problem by brute-force. Otherwise, we check
    whether $\Phi_k$ is trivially false or trivially true.\footnote{This step is the
    reason why we only get fixed-parameter tractability and not necessarily
    polynomial-time tractability.} If $\Phi_k$ is false, we output $0$; otherwise we output
    $\binom{n}{k}$. It is immediate that this algorithm yields fixed-parameter tractability.

    Now assume that $\Phi_k$ is non-trivial for infinitely many $k$. Since for $\Phi_k$ we fix the number of vertices to be $k$, by assumption $\Phi_k$ only
    depends on the number of edges of a graph. Thus, we have
    \begin{equation}\label{eq:edges_complement}
        \mathsf{hw}(f^{\neg\Phi,k}) = \binom{k}{2} - \mathsf{hw}(f^{\Phi,k}).
    \end{equation}
    Hence, set
    \[\hat{\Phi}_k := \begin{cases} \Phi_k &\text{if } \mathsf{hw}(f^{\Phi,k}) \leq \frac{1}{2}\binom{k}{2}\\
            \neg\Phi_k &\text{if } \mathsf{hw}(f^{\Phi,k}) > \frac{1}{2}\binom{k}{2} \,.
    \end{cases} \]
    We observe that, by assumption, $\mathcal{K}(\hat{\Phi})$ is infinite, and by
    \cref{fac:invariance} the problems $\#\indsubsprob(\Phi)$ and
    $\#\indsubsprob(\hat{\Phi})$ are equivalent. By definition and
    by~\cref{eq:edges_complement}, we see that $\mathsf{hw}(f^{\hat{\Phi},k}) \leq
    \frac{1}{2}\binom{k}{2}$ and therefore $\beta(k) = \binom{k}{2} -
    \mathsf{hw}(f^{\hat{\Phi},k}) \in \Theta(k^2)$. Thus, we have
    \[o\left(\frac{\beta(k)/k}{\log(\beta(k)/k)}\right) = o\left(k/\log k\right).\]
    The claim now follows by \cref{thm:main_general}.
\end{proof}

\subsection{Monotone Graph Properties}

Recall that a property $\Phi$ is called monotone if it is closed under taking subgraphs.
The decision version of $\#\indsubsprob(\Phi)$, that is, deciding whether there is an
induced subgraph of size $k$ that satisfies $\Phi$ is known to be $\W{1}$-complete if
$\Phi$ is monotone. Further, $\Phi$ is  non-trivial and $\mathcal{K}(\Phi)$ is infinite
(this follows implicitly by a result of Khot and Raman~\cite{KhotR02}).
However, as the  reduction of Khot and Raman is not parsimonious, the reduction does not
yield $\#\W{1}$-completeness of the counting version. More
importantly, the proof of Khot and Raman uses Ramsey's Theorem and thus only implies a
conditional lower bound of $g(k)\cdot |V(G)|^{o(\log k)}$. Using our main result, we achieve a
much stronger and almost tight lower bound under ETH.

\begin{theorem}\label{thm:monotone_basic}
    Let $\Phi$ denote a computable graph property that is monotone and non-trivial.
    Suppose that $\mathcal{K}(\Phi)$ is infinite. Then $\#\indsubsprob(\Phi)$ is
    $\#\W{1}$-complete and cannot be solved in time
    \[g(k)\cdot |V(G)|^{o\left(k/\log k\right)}\] for any function $g$, unless ETH fails. The
    same is true for $\#\indsubsprob(\overline{\Phi})$ and $\#\indsubsprob(\neg\Phi)$.
\end{theorem}
\begin{proof}
    As $\Phi$ is non-trivial, there is a graph $F$ such that $\Phi(H)$ is false for
    every $H$ that contains $F$ as a (not necessarily induced) subgraph. Set $r=|V(F)|$
    and fix $k\in \mathcal{K}(\Phi)$. By Tur\'ans Theorem (\cref{thm:turan}) we
    have that every graph $H$ on $k$ vertices with more than
    $\left(1-\frac{1}{r}\right)\cdot \frac{k^2}{2}$ edges contains the clique $K_{r+1}$
    and thus $F$ as a subgraph. Consequently, $\Phi$ is false on every graph with $k$
    vertices and more than $\left(1-\frac{1}{r}\right)\cdot \frac{k^2}{2}$ edges.
    Therefore, we have
    \[\beta(k)= \binom{k}{2} - \mathsf{hw}(f^{\Phi,k}) \geq \binom{k}{2} -
    \left(1-\frac{1}{r}\right)\cdot \frac{k^2}{2} = \frac{k^2}{2r} - \frac{k}{2} \in
\Omega(k^2). \]
    Thus $\beta(k) \in \Theta(k^2)$ and we conclude that
    \[o\left(\frac{\beta(k)/k}{\log(\beta(k)/k)}\right) = o\left(k/\log k\right) .\]
    The claim hence follows by \cref{thm:main_general}.
\end{proof}

\section{Refined Lower Bounds and Clique-Minors}\label{sec:refined_bounds}
Recall that the lower bounds of the previous section become tight if it is impossible to
``beat treewidth'', that is, if the $(\log k)^{-1}$ factor in the exponent of
\cref{thm:homsdicho} can be dropped. In this section, we show that the lower
bounds of the previous section can also be refined---and in case of sparse properties
even be made tight---without the latter assumption. This requires relying on two
results from extremal graph theory on forbidden cliques and clique-minors. The first one
is Tur\'an's Theorem, which we have seen already. The second one is a consequence of the
Kostochka-Thomason-Theorem:

\begin{theorem}[\cite{Kostochka84,Thomason01}]\label{thm:clique_minors}
    There is a constant $c > 0$ such that every graph $H$ with an average degree of at least
    $ct\sqrt{\log t}$ contains the clique $K_t$ as a minor.
\end{theorem}
Note that \cref{thm:clique_minors} is often stated in terms of the number of edges of a graph,
instead of its average degree. However, due to the Handshaking-Lemma, both statements are
equivalent.

Roughly speaking, we combine the Kostochka-Thomason-Theorem with
\cref{thm:main_theorem_combinatorial} to find graphs with large clique-minors in the
linear combination of homomorphisms associated with a graph property~$\Phi$ as given by
\cref{eq:indsubgmp}. This then allows us to derive hardness by a reduction
from the problem of \emph{finding} cliques, instead of relying on
\cref{thm:homsdicho}.

In what follows, we say that a subset $\mathcal{K}$ of the natural numbers is \emph{dense}
if there is a constant $\ell\geq 1$ such that for all $n\in \mathbb{N}_{>0}$ there
is a $k\in \mathcal{K}$ that satisfies $n\leq k \leq \ell n$. Now recall that
$\mathcal{K}(\Phi)$ is the set of all~$k$ such that $\Phi_k$ is not empty. The lower
bounds for $\#\indsubsprob(\Phi)$ in the current section require the set~$\mathcal{K}(\Phi)$ to be dense. Roughly speaking, this is to exclude artificial properties (such as for instance  $\Phi(H)= 1$ if and only if $H$ is an independent set and
$|V(H)|=2\uparrow m$ for some $m\in\mathbb{N}$, where $2\uparrow m$ is the $m$-fold
exponential tower of $2$). While the latter property satisfies the conditions of
\cref{thm:main_general}, we cannot construct a \emph{tight} reduction to
$\#\indsubsprob(\Phi)$ as the only non-trivial oracle queries satisfy $k= 2\uparrow m$ for
some $m\in\mathbb{N}$.
Fortunately, monotone properties exclude such artificial properties:
\begin{lemma}\label{lem:easy_monotone}
    Let $\Phi$ denote a non-trivial monotone graph property such that $\mathcal{K}(\Phi)$
    is infinite. Then~$\mathcal{K}(\Phi)$ is the set of all positive integers and thus
    dense.
\end{lemma}
\begin{proof}
    Fix a $n\in \mathbb{N}_{>0}$. As $\mathcal{K}(\Phi)$ is infinite, there is a $k\geq n$
    in $\mathcal{K}(\Phi)$. Thus there is a graph $H\in \Phi_k$. Now delete $k-n$
    arbitrary vertices of $H$ and call the resulting graph $H'$. As $\Phi$ is monotone, we
    have that $H'\in \Phi_n$ and hence $n\in \mathcal{K}(\Phi)$ as well.
\end{proof}

The following technical lemma is the basis for the lower bounds in this section; recall
that the problem $\#\homsprob(\mathcal{H})$ asks, given a graph $H\in \mathcal{H}$ and an
arbitrary graph $G$, to compute the number of homomorphisms from $H$ to $G$.
\begin{lemma}\label{lem:main_refined_bounds}
    Let $r\geq 1$ denote a constant and let $\mathcal{H}$ denote a decidable class of graphs.
    Further, let~$\mathcal{K}(\mathcal{H})$ denote the set of all positive integers $k$
    such that there is a graph $H\in \mathcal{H}$ with $k$ vertices and at least $k^2/2r -
    k/2$ edges.
    If $\mathcal{K}(\mathcal{H})$ is dense, then $\#\homsprob(\mathcal{H})$ cannot be
    solved in time
	\[ f(|V(H)|) \cdot |V(G)|^{o(|V(H)|/\sqrt{\log |V(H)|})} \]
	for any function $f$, unless ETH fails.
\end{lemma}
\begin{proof}
    We construct a tight reduction from the problem $\textsc{Clique}$, which asks, given a
    graph $G$ and a parameter $\hat{k}\in\mathbb{N}_{>0}$, to \emph{decide} whether there
    is a clique of size $\hat{k}$ in $G$. It is known that $\textsc{Clique}$ cannot be
    solved in time $\hat{f}(\hat{k})\cdot |V(G)|^{o(\hat{k})}$ for any function $\hat{f}$,
    unless ETH fails~\cite{Chenetal05,Chenetal06}.

    Now assume there is an algorithm $\mathbb{A}$ that solves $\#\indsubsprob(\Phi)$
    in time \[ f(|V(H)|) \cdot |V(G)|^{o(|V(H)|/\sqrt{\log |V(H)|})} \]
    for some function $f$. We use $\mathbb{A}$ to solve $\textsc{Clique}$ in time
    $\hat{f}(\hat{k})\cdot |V(G)|^{o(\hat{k})}$ for some function $\hat{f}$.

    As $\mathcal{K}(\mathcal{H})$ is dense, there is a constant $\ell \geq 1$ such
    that for all $\hat{k}\in \mathbb{N}_{>0}$ there is a $k\in
    \mathcal{K}(\mathcal{H})$ such that $\hat{k} \leq k \leq \ell \hat{k}$.
    Given a graph $G$ with $n$ vertices and $\hat{k}\in \mathbb{N}_{>0}$, we construct the
    graph $G'$ as follows: We first search for a $k\in \mathcal{K}(\mathcal{H})$
    that satisfies $\hat{k} \leq k \leq \ell \hat{k}$. Note that finding $k$ is computable as
    $\mathcal{K}(\mathcal{H})$ is decidable. Then, we add $k-\hat{k}$ ``fresh'' vertices
    to $G$ and add edges between all pairs of new vertices and between all pairs of a new
    vertex and an old vertex.

    Next, let $c$ denote
    the constant from \cref{thm:clique_minors} and set
    \[h(k):= cr'\cdot \sqrt{\log ({k}/{cr'})}\,,\]
	where $r' := \max\{1/c,r+1\}$.
    Now, we construct a graph $\hat{G}$ from $G'$ as follows:
    The vertices of $\hat{G}$ are the $h(k)$-cliques of $G'$,\footnote{\label{ftnt:rounding}Formally, we have to round $h(k)$ and later $k/h(k)$. For the sake of readability, we assume that all logs and fractions yield integers, but we point out that this might require to find a $\lceil k/h(k)\rceil$-clique at the end of the proof, while the oracle can only determine the existence of a $\lfloor k/h(k)\rfloor$-clique. However, using the latter, we can easily decide whether there is a $\lceil k/h(k)\rceil$-clique  by checking for each vertex whether its neighbourhood contains a $\lfloor k/h(k)\rfloor$-clique.} and two vertices $C_1$ and
    $C_2$ of $\hat{G}$ are made adjacent if all edges between vertices in $C_1$ and
    vertices in $C_2$ are present in $G'$. Note that $\hat{G}$ has $O(n^{h(k)})$
    vertices and can be constructed in time $g(\hat{k})\cdot O(n^{h(k)})$ for some
    computable function $g$.
	\begin{claim}\label{clm:cliqueETH}
        The graph $G$ contains a clique of size $\hat{k}$ if and only if $\hat{G}$
        contains a clique  of size $k/h(k)$.
	\end{claim}
	\begin{claimproof}
        Let $C$ denote a $\hat{k}$-clique in $G$. Then we obtain a $k$-clique in $G'$ by adding
        the fresh $k-\hat{k}$ vertices to $C$. Next, partition $C$ in blocks of size $h(k)$.
        Each block will be a vertex of $\hat{G}$ and all corresponding vertices are pairwise
        adjacent by the construction of $\hat{G}$. As there are $k/h(k)$ many blocks, we found
        the desired clique in $\hat{G}$.

        For the other direction, let $\hat{C}$ denote a $k/h(k)$-clique in $\hat{G}$. By the
        definition of $\hat{G}$, each vertex of $\hat{C}$ corresponds to a clique of size
        $h(k)$ in $G'$. Furthermore, two different of those cliques cannot share a common
        vertex as the corresponding vertices in $\hat{G}$ are adjacent---recall that we do
        not allow self-loops. Consequently, the union of the $k/h(k)$ many cliques constitutes
        a clique of size $k$ in $G'$. Finally, at most $k-\hat{k}$ of the vertices of this
        clique can be fresh vertices, and thus $G$ contains a clique of size (at least)
        $\hat{k}$.
    \end{claimproof}
    Next, we search for a graph $H\in \mathcal{H}$ with $k$ vertices and at least
    $k^2/2r - k/2$ edges. By assumption, this can be done in time $g'(k)$ for some
    computable function $g$ as $\mathcal{H}$ is decidable. Now we see that
    \[d(H)= \frac{1}{k} \cdot \sum_{v \in V(H)} \mathsf{deg}(v) = \frac{2 |E(H)|}{k} \geq
    \frac{k}{r} -1\,. \]
    Set $t = \frac{k}{h(k)}$. For $k$ large enough,\footnote{If $k$ is not large enough
    for the inequalities to hold, then $k$ and thus $\hat{k}$ are bounded by a constant,
    and we can compute the number of $\hat{k}$-cliques in $G$ by brute-force.} we have
    \begin{align*}
        ct\sqrt{\log t} &= c \cdot  \frac{k}{cr'\cdot \sqrt{\log
                ({k}/{cr'})}} \cdot \sqrt{\log ({k}/{cr') -
                \log\sqrt{\log ({k}/{cr'})}}}\\
                &\leq c \cdot  \frac{k}{cr'\cdot \sqrt{\log ({k}/{cr'})}}
                \cdot \sqrt{\log ({k}/{cr'})}
                = \frac{k}{r'} \leq \frac{k}{r} - 1 \leq d(H)
    \end{align*}
    \Cref{thm:clique_minors} thus implies that $K_t$ is a minor of $H$. Furthermore, it is
    known that, whenever a graph $F$ is a minor of a graph $H$, there is a tight reduction
    from counting homomorphisms from $F$ to counting homomorphisms from $H$ --- see for
    instance~\cite[Chapter~2.5]{Roth19} and Section~3 in the full version\footnote{Full version available at \url{https://arxiv.org/abs/1902.04960}.} of~\cite{DellRW19icalp}. Consequently, we can use the algorithm $\mathbb{A}$ to
    compute the number of homomorphisms from $K_t$ to $\hat{G}$. By assumption on
    $\mathbb{A}$, this takes time at most
    \[ f(|V(H)|) \cdot |V(\hat{G})|^{o(|V(H)|/\sqrt{\log |V(H)|})}  = f(k) \cdot
    (n^{h(k)})^{o(k/\sqrt{\log k})} = f(k) \cdot n^{o(k)} \,,\]
    where the latter holds as $h(k)\in \Theta(\sqrt{\log k})$---recall that $r'$ and $c$
    are constants. However, it is easy to see that $\hat{G}$ contains a $t$-clique if and
    only if the number of homomorphisms from $K_t$ to $\hat{G}$ is at least $1$. By
    \cref{clm:cliqueETH}, this is equivalent to $G$ having a clique of size
    $\hat{k}$. As $k \in O(\hat{k})$ (recall that $\ell$ is a constant), the overall
    running time is bounded by
    \[ g(\hat{k})\cdot O(n^{h(k)}) + g'(k) + f(k) \cdot n^{o(k)} \leq \hat{f}(\hat{k})
    \cdot n^{o(\hat{k})}   \]
    for $\hat{f}(\hat{k}) := g(\hat{k}) + g'(\ell \hat{k}) + f(\ell \hat{k})$. This yields
    the desired contradiction and concludes the proof.
\end{proof}

We are now able to establish the refined lower bounds for $\#\indsubsprob(\Phi)$.
\begin{theorem}\label{thm:monotone_refined}
    Let $\Phi$ denote a computable graph property that is monotone and non-trivial.
    Suppose that $\mathcal{K}(\Phi)$ is infinite. Then $\#\indsubsprob(\Phi)$ cannot be
    solved in time
    \[g(k)\cdot |V(G)|^{o(k/\sqrt{\log k})}\] for any function $g$, unless ETH
    fails. The
    same is true for $\#\indsubsprob(\overline{\Phi})$ and $\#\indsubsprob(\neg\Phi)$.
\end{theorem}
\begin{proof}
    We begin similarly as in the proof of \cref{thm:monotone_basic}: As $\Phi$ is
    non-trivial, there is a graph $F$ such that $\Phi(H)$ is false for
    every $H$ that contains $F$ as a (not necessarily induced) subgraph. Set $r=|V(F)|$
    and fix $k\in \mathcal{K}(\Phi)$. By Tur\'ans Theorem (\cref{thm:turan}) we
    have that every graph $H$ on $k$ vertices with more than
    $\left(1-\frac{1}{r}\right)\cdot \frac{k^2}{2}$ edges contains the clique $K_{r+1}$
    and thus $F$ as a subgraph. Consequently, $\Phi$ is false on every graph with $k$
    vertices and more than $\left(1-\frac{1}{r}\right)\cdot \frac{k^2}{2}$ edges.
    We now use \cref{thm:main_theorem_combinatorial} and obtain a computable and unique
    function $a$ of finite support such that
    \[\#\indsubs{\Phi,k}{\ast} = \sum_{H \in \mathcal{G}} a(H) \cdot \#\homs{H}{\ast},\]
    satisfying that there is a graph $H_k$ on $k$ vertices and at least
    \[\binom{k}{2}-\mathsf{hw}(f^{\Phi,k})+1  \geq \binom{k}{2} -
    \left(1-\frac{1}{r}\right)\cdot \frac{k^2}{2} + 1 \geq \frac{k^2}{2r} - \frac{k}{2}\]
    edges such that $a(H_k) \neq 0$.
    Consequently, Complexity Monotonicity (\cref{thm:monotonicity}) yields a tight
    reduction from the problem $\#\homsprob(\mathcal{H})$ where $\mathcal{H}:=\{H_k~|~
    k\in \mathcal{K}(\Phi)\}$. By the previous observation, the graph $H_k$ has $k$
    vertices and at least $k^2/2r - k/2$ many edges, and by \cref{lem:easy_monotone}
    the set of $k$ such that $H_k\in \mathcal{H}$ is dense. Thus we can use
    \cref{lem:main_refined_bounds}, which concludes the proof---note that the
    results for $\#\indsubsprob(\overline{\Phi})$ and $\#\indsubsprob(\neg\Phi)$ follow by
    \cref{fac:invariance}.
\end{proof}

We continue with the refined lower bound for properties that only depend on the number of
edges of a graph; in this case, we have to assume density.
\begin{theorem}\label{cor:number of edges_refined}
	Let $\Phi$ denote a computable graph property that only depends on the number of edges of a
    graph. If the set of $k$ for which $\Phi_k$ is non-trivial is dense, then
    $\#\indsubsprob(\Phi)$ cannot be solved in time \[g(k)\cdot
    |V(G)|^{o(k/\sqrt{\log k})}\] for any function $g$,
	unless ETH fails.
\end{theorem}
\begin{proof}
    We use the same set-up as in the proof of \cref{cor:number of edges}. In
    particular, we obtain $\hat{\Phi}$ such that $\#\indsubsprob(\Phi)$ and
    $\#\indsubsprob(\hat{\Phi})$ are equivalent, $\mathcal{K}(\hat{\Phi})$ is dense, and
    $\mathsf{hw}(f^{\hat{\Phi},k}) \leq \frac{1}{2}\binom{k}{2}$ for all $k\in
    \mathcal{H}(\hat{\Phi})$. We use \cref{thm:main_theorem_combinatorial} and
    obtain a computable and unique function $a$ of finite support such that
	\[\#\indsubs{\Phi,k}{\ast} = \sum_{H \in \mathcal{G}} a(H) \cdot \#\homs{H}{\ast},\]
    satisfying that there is a graph $H_k$ on $k$ vertices and at least
    \[\binom{k}{2}-\mathsf{hw}(f^{\Phi,k})+1  \geq \frac{k^2}{4} - \frac{k}{2}\]
    edges such that $a(H_k) \neq 0$. The application of Complexity Monotonicity
    (\ref{thm:monotonicity}) and \cref{lem:main_refined_bounds} is now similar to the
    previous proof; the only difference is, that we can choose $r=2$.
\end{proof}

As a final result in this section, we establish a tight conditional lower bound for sparse
properties; recall that a property $\Phi$ is called \emph{sparse} if there is a
constant $s$ only depending on $\Phi$ such that $\Phi$ is false on every graph with $k$
vertices and more than $sk$ edges. Instead of relying in the Kostochka-Thomason-Theorem,
it suffices to use Tur\'an's Theorem, however.

\begin{theorem}\label{thm:sparse_tight}
    Let $\Phi$ denote a computable sparse graph property such that $\mathcal{K}(\Phi)$ is
    dense. Then $\#\indsubsprob(\Phi)$ cannot be solved in time \[g(k)\cdot
    |V(G)|^{o\left(k\right)}\] for any function $g$,
    unless ETH fails. The same is true for $\#\indsubsprob(\overline{\Phi})$ and
    $\#\indsubsprob(\neg\Phi)$.
\end{theorem}
\begin{proof}
    Let $s$ denote the constant given by the definition of sparsity, and fix $k\in
    \mathcal{K}(\Phi)$. We use \cref{thm:main_theorem_combinatorial} and obtain
    a computable and unique function $a$ of finite support such that
	\[\#\indsubs{\Phi,k}{\ast} = \sum_{H \in \mathcal{G}} a(H) \cdot \#\homs{H}{\ast},\]
    satisfying that there is a graph $H_k$ on $k$ vertices and at least
    $\binom{k}{2}-\mathsf{hw}(f^{\Phi,k})+1$ edges such that $a(H_k) \neq 0$.
	 Now choose $r := \lceil\frac{k}{2(s+1)+1}\rceil$,
	and observe that
    \[\binom{k}{2}-\mathsf{hw}(f^{\Phi,k})+1 > \binom{k}{2}-sk >  \binom{k}{2}-(s+1)k \geq
    \left(1-\frac{1}{r}\right)\cdot \frac{1}{2}|V(H_k)|^2\,.\]
    By Tur\'ans Theorem (\cref{thm:turan}), $H_k$ hence contains $K_{r+1}$ as a subgraph
    and, in particular, $K_r$ as a minor. Furthermore, Complexity Monotonicity
    (\cref{thm:monotonicity}) shows that we can, given a graph $G$ compute
    $\#\homs{H_k}{G}$ in linear time if we are given oracle access to
    $\#\indsubs{\Phi,k}{\star}$. As we have seen in the proof of Lemma~\ref{lem:main_refined_bounds}, it is known that, whenever a graph $F$ is a
    minor of a graph $H$, there is a tight reduction from counting homomorphisms from $F$
    to counting homomorphisms from $H$~\cite{DellRW19icalp,Roth19}. In particular,
    this implies that we can compute $\#\homs{K_{r}}{G}$ in linear time if we are given
    oracle access to $\#\indsubs{\Phi,k}{\star}$. Note further that $\#\homs{K_{r}}{G}$
    is at least~$1$ if and only if $G$ contains a clique of size $r$.

	We continue similarly as in the proof of \cref{lem:main_refined_bounds}:
    Assume that there is an algorithm $\mathbb{A}$ that solves $\#\indsubsprob(\Phi)$ in
    time $g(k)\cdot |V(G)|^{o\left(k\right)}$ for some function $g$. We show that
    $\mathbb{A}$ can be used to solve the problem of \emph{finding} a clique of size
    $\hat{k}$ in a graph $G$ in time $\hat{f}(\hat{k})\cdot|V(G)|^{o(\hat{k})}$
    for some function $\hat{f}$. Given $\hat{k}$ and $G$, search for the smallest
    $k\in\mathcal{K}(\Phi)$ such that
	\[\hat{k} \leq \lceil\frac{k}{2(s+1)+1}\rceil \,,\]
    and note that finding such a $k$ is computable in time only depending on $\hat{k}$ as
    $\Phi$ is computable. Note further that $k\in O(\hat{k})$ as $s$ is a constant and
    $\mathcal{K}(\Phi)$ is dense.

    We construct the graph $G'$ from $G$ by adding $\lceil\frac{k}{2(s+1)+1}\rceil - \hat{k}$ ``fresh''
    vertices and adding edges between every pair of new vertices and every pair containing
    one new and one old vertex. It is easy to see that $G$ has a clique of size $\hat{k}$
    if and only if $G'$ has a clique of size $\lceil\frac{k}{2(s+1)+1}\rceil$.
    By the analysis and
    assumptions above, we can decide whether the latter is true by using $\mathbb{A}$ in
    time $g(k)\cdot |V(G')|^{o\left(k\right)}$. As $|V(G')|\in O(|V(G)|)$ and $k\in
    O(\hat{k})$, the overall time to decide whether $G$ has a clique of size $\hat{k}$ is
    hence bounded by
		$\hat{f}(\hat{k})\cdot|V(G)|^{o(\hat{k})}$
    for some function $\hat{f}$, which is impossible, unless ETH
    fails~\cite{Chenetal05,Chenetal06}. The results for $\#\indsubsprob(\overline{\Phi})$
    and $\#\indsubsprob(\neg\Phi)$ follow by Fact~\ref{fac:invariance}.
\end{proof}

\section{Hereditary Graph Properties}\label{sec:hereditary}
While we established hardness for a variety of graph properties with $f$-vectors of small
hamming weight in the previous sections, we observe that our meta-theorem
(\cref{thm:main_general}) does not apply for properties~$\Phi$ for which
$f^{\Phi,k}$, $f^{\neg\Phi,k}$ and $f^{\overline{\Phi},k}$ have large hamming weight.
A well-studied class of properties containing examples of such $\Phi$ is the family of
hereditary graph properties: In contrast to monotone properties, which are closed under
taking subgraphs, a property $\Phi$ is called \emph{hereditary} if it is closed under
taking \emph{induced} subgraphs. It is a well-known fact that every hereditary property
$\Phi$ is characterized by a (possibly infinite) set $\Gamma(\Phi)$ of forbidden induced
subgraphs, that is
\[\Phi(G) = 1 \Leftrightarrow \forall H \in \Gamma(\Phi): \#\indsubs{H}{G}=0\,. \]
Given any graph $H$, the property $\Phi$ of being \emph{$H$-free}, that is, not containing
$H$ as an induced subgraph, is hereditary with  $\Gamma(\Phi)=\{H\}$.

In this section, we settle the hardness-question for many hereditary graph properties as
well. Note that $\Phi$ is hereditary if and only if its inverse $\overline{\Phi}$ is. In
particular, $H\in \Gamma(\Phi) \Leftrightarrow \overline{H} \in \Gamma(\overline{\Phi})$.
\begin{restatable}{theorem}{hermn}
    Let $H$ denote graph that is not the trivial graph with a single vertex and let $\Phi$
    denote the property $\Phi(G)=1:\Leftrightarrow$ $G$ is $H$-free. Then
    $\#\indsubsprob(\Phi)$ is $\#\W{1}$-complete and cannot be solved in time
	\[g(k)\cdot |V(G)|^{o(k)}\] for any function $g$, unless ETH fails. The same is true for
    the problem $\#\indsubsprob(\neg\Phi)$.\ifx\hermnt\undefined{\lipicsEnd}\fi
\end{restatable}
\def\hermnt{1}

\noindent We start with some terminology used in this section.
Given a graph $H$, a pair $(u,v) \in V(H)^2$, and two non-negative integers $x$ and $y$, we
construct the \emph{exploded} graph $H(u,v,x,y)$ by adding $x-1$ clones of $u$ and $y-1$
clones of $v$, including all incident edges; if $x$ or $y$ are zero, then we delete $u$ or
$v$, respectively. Given an edge $e=\{u,v\}$ of $H$ and two non-negative integers $x$ and
$y$, we define the $e$-\emph{exploded} graph as $H_{u,v}^{x,y} := (V(H),E(H)\setminus \{u,
v\})(u,v,x,y)$. Consult \cref{fig:explosion} for a visualization. An edge $e=\{u,v\}$ of a
graph $H$ is called \emph{critical}, if $\#\indsubs{H}{H^{x,y}_{u,v}}=0$ for every
pair~$x,y\in \mathbb{N}_{\geq 0}$.

Now let $\Phi$ denote a hereditary graph property and let
$\Gamma(\Phi)$ denote the associated set of forbidden induced subgraphs. We say that $\Phi$
has a \emph{critical edge} if there is a graph $H\in \Gamma(\Phi)$ and an edge $\{u,v\}\in
E(H)$ such that  for all positive integers $x$ and
$y$, the graph $H_{u,v}^{x,y}$ satisfies $\Phi$, that is, for every $\hat{H}\in\Gamma(\Phi)$,
we have \[\#\indsubs{\hat{H}}{H_{u,v}^{x,y}} = 0.\]
Finally, we say that a hereditary property $\Phi$ is \emph{critical} if either $\Phi$ or
its inverse $\overline{\Phi}$ has a critical edge. We will see later in this section that
every critical property $\Phi$ will induce hardness of $\#\indsubsprob(\Phi)$.

\usetikzlibrary{calc,shapes,fit}

\tikzset{vertex/.style={circle, fill, inner sep=1.5pt, outer sep=1.9pt}}
\tikzset{cvertex/.style={circle, fill, inner sep=1pt, outer sep=1.3pt, lipicsGray!60}}
\tikzset{ccvertex/.style={circle, fill, inner sep=1.3pt, outer sep=1.7pt, lipicsGray!80}}
\tikzset{svertex/.style={circle,draw=white, line width=1.2pt,fill=black, inner sep=1.8pt, outer         sep=0pt}}

\def\ccol{red}
\def\ecol{blue}
\def\icol{red!50!blue}

\tikzset{edge/.style={very thick}}
\tikzset{ccedge/.style={lipicsGray!50}}
\tikzset{cedge/.style={lipicsGray!75}}
\tikzset{dedge/.style={gray,thick, double=white, double distance=9pt}}
\tikzset{sdedge/.style={white,thick, double=black, double distance=1.2pt}}

\begin{figure}[t]
    \centering
    \begin{tikzpicture}
        \begin{scope}
            \node[vertex,label=right:{$u$}] (u) at (0,0) {};
            \node[vertex,label=right:{$v$}] (v) at (0,-2) {};

            \draw[edge] (u) -- (v);

            \node (e1) at (-2.2, .1) {};
            \node (e2) at (.4, -2.1) {};
            \node[cvertex] (x1) at (-1.8, -1.5) {};
            \node[cvertex] (x2) at (-2, -.5) {};
            \node[cvertex] (x3) at (-1.6, -1) {};
            \node[fit=(x1)(x2)(x3), rounded corners=3pt, draw=lipicsGray!40,
                fill=lipicsGray!10] {};
            \node[fit=(e1)(e2), rounded corners=5pt, draw] (f) {};
            \node at ($(f.south) + (0,-.3)$) {$H$};

            \node[cvertex] (x1) at (-1.7, -1.5) {};
            \node[cvertex] (x2) at (-2, -.5) {};
            \node[cvertex] (x3) at (-1.6, -.84) {};
            \node[cvertex] (x4) at (-2, -1.16) {};
            \draw[cedge] (u) -- (x3) -- (v);
            \draw[cedge] (u) -- (x2);
            \draw[cedge] (v) -- (x1) -- (x2) -- (x3) -- (x4);
        \end{scope}
        \begin{scope}[xshift=12em]
            \node[vertex,label=right:{$u$}] (u) at (0,0) {};
            \node[ccvertex,label=right:{$u$}] (u1) at (0,-.4) {};
            \node[ccvertex,label=right:{$u$}] (u2) at (0,-.8) {};

            \node[ccvertex,label=right:{$v$}] (v1) at (0,-1.6) {};
            \node[vertex,label=right:{$v$}] (v) at (0,-2) {};

            \node (e1) at (-2.2, .1) {};
            \node (e2) at (.4, -2.1) {};
            \node[cvertex] (x1) at (-1.8, -1.5) {};
            \node[cvertex] (x2) at (-2, -.5) {};
            \node[cvertex] (x3) at (-1.6, -1) {};
            \node[fit=(x1)(x2)(x3), rounded corners=3pt, draw=lipicsGray!40,
                fill=lipicsGray!10] {};
            \node[fit=(e1)(e2), rounded corners=5pt, draw] (f) {};
            \node at ($(f.south) + (0,-.3)$) {$H^{3,2}_{u, v}$};

            \node[cvertex] (x1) at (-1.7, -1.5) {};
            \node[cvertex] (x2) at (-2, -.5) {};
            \node[cvertex] (x3) at (-1.6, -.84) {};
            \node[cvertex] (x4) at (-2, -1.16) {};
            \draw[ccedge] (u1) -- (x3) -- (v1);
            \draw[ccedge] (u2) -- (x3);
            \draw[ccedge] (u1) -- (x2);
            \draw[ccedge] (u2) -- (x2);
            \draw[ccedge] (v1) -- (x1);
            \draw[cedge] (u) -- (x3) -- (v);
            \draw[cedge] (u) -- (x2);
            \draw[cedge] (v) -- (x1) -- (x2) -- (x3) -- (x4);
        \end{scope}
        \begin{scope}[xshift=24em]
            \node[ccvertex,label=right:{$v$}] (v1) at (0,-1.6) {};
            \node[ccvertex,label=right:{$v$}] (v2) at (0,-1.2) {};
            \node[vertex,label=right:{$v$}] (v) at (0,-2) {};

            \node (e1) at (-2.2, .1) {};
            \node (e2) at (.4, -2.1) {};
            \node[cvertex] (x1) at (-1.8, -1.5) {};
            \node[cvertex] (x2) at (-2, -.5) {};
            \node[cvertex] (x3) at (-1.6, -1) {};
            \node[fit=(x1)(x2)(x3), rounded corners=3pt, draw=lipicsGray!40,
                fill=lipicsGray!10] {};
            \node[fit=(e1)(e2), rounded corners=5pt, draw] (f) {};
            \node at ($(f.south) + (0,-.3)$) {$H^{0, 3}_{u, v}$};

            \node[cvertex] (x1) at (-1.7, -1.5) {};
            \node[cvertex] (x2) at (-2, -.5) {};
            \node[cvertex] (x3) at (-1.6, -.84) {};
            \node[cvertex] (x4) at (-2, -1.16) {};
            \draw[ccedge] (v2) -- (x3) -- (v1);
            \draw[ccedge] (v1) -- (x1) -- (v2);
            \draw[cedge] (x3) -- (v);
            \draw[cedge] (v) -- (x1) -- (x2) -- (x3) -- (x4);
        \end{scope}
    \end{tikzpicture}
    \caption{Different explosions of the edge $\{u, v\}$ of  a graph $H$.}\label{fig:explosion}
\end{figure}
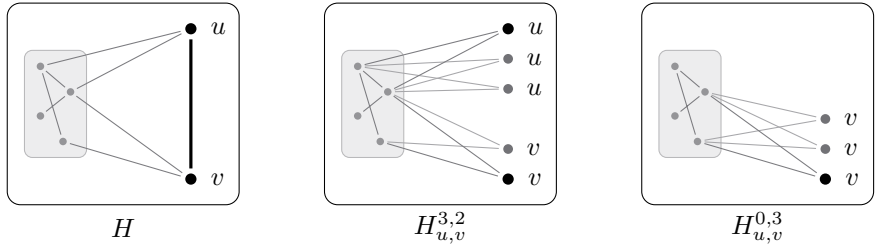

First, we establish that for every graph $H$ with at least two vertices, the property
``$H$-free'' is critical. To this end, we rely on neighbour-sharing vertices: Given
two vertices $u$ and $v$ of a graph $H$, we say that $u$ and $v$ are \emph{false twins} if
they have the same set of adjacent vertices. Note that, in particular, false twins cannot
be adjacent as we consider graphs without self-loops. Furthermore, for a graph $H$, we
define the partition $P(H)$ by adding two vertices to the same block if and only if they
are false twins. Finally, we define the graph $H\!\!\downarrow$ by identifying all vertices in
$H$ with the block of $P(H)$ they belong to, that is, the vertices of $H\!\!\downarrow$
are the blocks of $P(H)$ and
two blocks $B$ and $B'$ are adjacent if there are vertices $v\in B$ and $v'\in B'$ such
that $\{v,v'\}\in E(H)$.

\begin{lemma}\label{lem:hfree_critical}
    Let $H$ denote a graph with at least $2$ vertices and let $\Phi$ denote a hereditary graph
    property such that~$\Gamma(\Phi)=\{H\}$. Then $\Phi$ is critical.
\end{lemma}
\begin{proof}
    We show that either $\Phi$ or $\overline{\Phi}$ has a critical edge. As
    $\Gamma(\Phi)=\{H\}$, we need to prove that at least one of $H$ and
    $\overline{H}$ has a critical edge. Let us start with the following claim.
    \begin{claim}\label{clm:critical_singleton}
        Let $F$ denote a graph with an edge $\{u,v\}$ such that $\{u\}$ and $\{v\}$ are
        singleton sets in the partition $P(F)$. Then $\{u,v\}$ is a critical edge of $F$.
    \end{claim}
    \begin{claimproof}
        Suppose there are integers $x,y\in \mathbb{N}_{\geq 0}$ such that there is an induced
        subgraph $F'$ of $F^{x,y}_{u,v}$ that is isomorphic to $F$. Then there is a
        natural bijection between the blocks of $F$ and the blocks of  $F^{x,y}_{u,v}$
        sending the class of each vertex not equal to $u,v$ to themselves, sending the
        class $\{u\}$ to the class of its $x$ clones and similar for $v$. But note that
        the graph $F^{x,y}_{u,v}\!\!\downarrow$ has one fewer edge than $F\!\!\downarrow$
        (since the previous edge $\{u,v\}$ was removed). However, the induced subgraph
        $F'$ of  $F^{x,y}_{u,v}$ satisfies that $F'\!\!\downarrow$ is a subgraph of
        $F^{x,y}_{u,v}\!\!\downarrow$ and thus also has at least one fewer edges than
        $F\!\!\downarrow$, a contradiction to $F'$ being isomorphic to $F$.
    \end{claimproof}
    Using the previous claim, it suffices to show that there are two vertices $u$ and $v$
    such that $\{u,v\}$ is an edge and $\{u\}$ and $\{v\}$ are singletons in either one of
    $H$ and $\overline{H}$.

    We first show that for every vertex $z\in V(H)=V(\overline{H})$, the set $\{z\}$ is either a
    singleton in $P(H)$ or in $P(\overline{H})$. To this end, assume that $z$ has a false
    twin $z'$ in $H$ and a false twin $z''$ in $\overline{H}$. Consequently, $\{z,z'\}\notin
    E(H)$ and $\{z,z''\}\notin E(\overline{H})$, and thus $\{z,z'\} \in E(\overline{H})$ and
    $\{z,z''\}\in E(H)$. Now, as $z$ and $z'$ are false twins in $H$ and $\{z,z''\}\in E(H)$,
    we see that $\{z', z''\} \in E(H)$. However, as $z$ and $z''$ are false twins in
    $\overline{H}$ and $\{z,z'\} \in E(\overline{H})$, we see that $\{z', z''\} \in
    E(\overline{H})$ as well, which leads to the desired contradiction.

    Now assume without loss of generality that $H$ has at least one edge; otherwise we consider
    $\overline{H}$. If there are false twins $u$ and $v$ in $H$, then, by the previous
    argument, $\{u\}$ and $\{v\}$ are singletons in $P(\overline{H})$ and $u$ and $v$ are
    adjacent in $\overline{H}$. By \cref{clm:critical_singleton}, we obtain a
    critical edge of $\overline{H}$. If there are no false twins in~$H$, we can choose an
    arbitrary edge of~$H$ which is then again critical by
    \cref{clm:critical_singleton}. This concludes the
    proof.\footnote{For your amusement: if you color all vertices red that are singletons in $H$
    and color all vertices blue that are singletons only in $\overline H$, then the fact
    that the Ramsey number $R(2)$ is $3$ shows that for $H$ having at least $3$
    vertices, we find a critical edge in $H$ or $\overline{H}$.}
\end{proof}
We have shown that every hereditary property defined by a single (non-trivial) forbidden
induced subgraph is critical. Let us now provide some examples of critical hereditary
properties that are defined by multiple forbidden induced subgraphs:
\begin{enumerate}
    \item $\Phi(H) = 1 :\Leftrightarrow H$ is perfect. A graph $H$ is \emph{perfect} if
        for every induced subgraph of $H$, the size of the largest clique equals the
        chromatic number. By the Strong Perfect Graph Theorem~\cite{ChudnovskyRST06}, we
        have that~$\Gamma(\Phi)$ is the set of all odd cycles of length at least $5$ and
        their complements. Now observe that every edge of the cycle of length $5$ is
        critical for $\Phi$ as the exploded graph is bipartite and thus perfect.
    \item $\Phi(H) = 1 :\Leftrightarrow H$ is chordal. A graph is \emph{chordal} if it
        does not contain an induced cycle of length $4$ or more. Consequently, we can
        choose an arbitrary edge of the cycle of length $4$ as a critical edge for $\Phi$,
        as the resulting exploded graphs do not contain any cycle.
    \item $\Phi(H) = 1 :\Leftrightarrow H$ is a split graph. A \emph{split graph} is a
        graph whose vertices can be partitioned into a clique and an independent set. It
        is known~\cite{FoldesH77} that $\Gamma(\Phi)$ contains the cycles of length $4$
        and $5$, and the complement of the cycle of length $4$. The latter is the graph
        containing two disjoint edges, and it is easy to see that any of those two edges
        is critical for $\Phi$.
\end{enumerate}

We now establish hardness of $\#\indsubsprob(\Phi)$ for critical hereditary properties. We
start with the following lemma, which constructs a (tight) parameterized Turing-reduction from
counting cliques of size~$k$ in bipartite graphs.
\begin{lemma}\label{lem:reduction_hereditary}
    Let $\Phi$ denote a computable and critical hereditary graph property. There is an
    algorithm~$\mathbb{A}$ with oracle access to $\#\indsubsprob(\Phi)$ that expects as
    input a bipartite graph $G$ and a positive integer $k$, and computes the number of
    independent sets of size $k$ in $G$ in time $O(|G|)$. Furthermore, the number of calls
    to the oracle is bounded by $O(1)$ and every queried pair~$(\hat{G},\hat{k})$
    satisfies $|V(\hat{G})| \in O(|V(G)|)$ and $\hat{k} \in O(k)$.
\end{lemma}
\begin{proof}
    Assume without loss of generality that $\Phi$ has a critical edge; otherwise we use
    \cref{fac:invariance} and proceed with $\overline{\Phi}$. Hence choose $H\in
    \Gamma(\Phi)$ and $e=\{u,v\} \in E(H)$ such that $\#\indsubs{\hat{H}}{H_{u,v}^{x,y}} = 0$
    for every $\hat{H}\in \Gamma(\Phi)$ and for all non-negative integers $x,y$.

    Now let $G=(U\dot\cup V, E)$ and $k$ denote the given input. If $U$ or $V$ are
    empty, then we can trivially compute the number of independent sets of size $k$.
    Hence, we have $U=\{u_1,\dots,u_{n_1}\}$ and $V=\{v_1,\dots,v_{n_2}\}$ for some
    integers $n_1,n_2 >0$. Now, we proceed as follows:
    In the first step, we construct the graph $H_{u,v}^{n_1,n_2}$.
    In the next step, we identify $u$ and its
    $n_1-1$ clones with the vertices of $U$ and $v$ and its $n_2-1$ clones with the
    vertices of $V$. Finally, we add the edges $E$ of $G$. We call the resulting graph
    $\hat{G}$ and we observe that $\hat{G}$ can clearly be constructed in time
    $O(|G|)$; note that $|H|$ is a constant as $\Phi$ is fixed. Consult \cref{fig:ghat}
    for a visualization of the construction.

    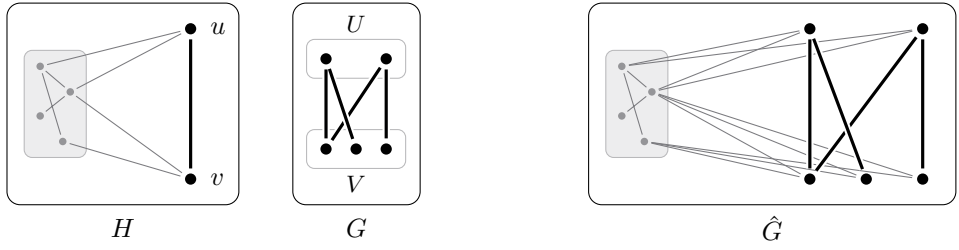
\begin{figure}[t]
        \centering
        \begin{tikzpicture}
            \begin{scope}
                \node[vertex,label=right:{$u$}] (u) at (0,0) {};
                \node[vertex,label=right:{$v$}] (v) at (0,-2) {};

                \draw[edge] (u) -- (v);

                \node (e1) at (-2.2, .1) {};
                \node (e2) at (.4, -2.1) {};
                \node[cvertex] (x1) at (-1.8, -1.5) {};
                \node[cvertex] (x2) at (-2, -.5) {};
                \node[cvertex] (x3) at (-1.6, -1) {};
                \node[fit=(x1)(x2)(x3), rounded corners=3pt, draw=lipicsGray!40,
                    fill=lipicsGray!10] {};
                \node[fit=(e1)(e2), rounded corners=5pt, draw] (f) {};
                \node at ($(f.south) + (0,-.3)$) {$H$};

                \node[cvertex] (x1) at (-1.7, -1.5) {};
                \node[cvertex] (x2) at (-2, -.5) {};
                \node[cvertex] (x3) at (-1.6, -.84) {};
                \node[cvertex] (x4) at (-2, -1.16) {};
                \draw[cedge] (u) -- (x3) -- (v);
                \draw[cedge] (u) -- (x2);
                \draw[cedge] (v) -- (x1) -- (x2) -- (x3) -- (x4);

                \node[vertex] (U3) at (2.6, -0.4) {};
                \node[vertex] (U4) at (1.8, -0.4) {};
                \node[vertex] (V2) at (2.6, -1.6) {};
                \node[vertex] (V3) at (2.2, -1.6) {};
                \node[vertex] (V4) at (1.8, -1.6) {};
                \node[fit=(U3)(U4), rounded corners=3pt, draw=lipicsGray!40] (x1) {};
                \node[fit=(V2)(V3)(V4), rounded corners=3pt, draw=lipicsGray!40] (x2) {};
                \node at ($(x2.south) + (0,-.2)$) {\small$V$};
                \node at ($(x1.north) + (0,.2)$) {\small$U$};
                \node (f1) at (2.81, .1) {};
                \node (f2) at (1.6, -2.1) {};
                \draw[sdedge] (U3) -- (V4);
                \draw[sdedge] (U3) -- (V2);
                \draw[sdedge] (V3) -- (U4) -- (V4);
                \node[fit=(f1)(f2), rounded corners=5pt, draw] (g) {};
                \node at ($(g.south) + (0,-.3)$) {$G$};
            \end{scope}
            \begin{scope}[xshift=22em]
                \node (e1) at (-2.2, .1) {};
                \node (e2) at (.4, -2.1) {};
                \node (f1) at (2.2, .1) {};
                \node (f2) at (1.6, -2.1) {};
                \node[cvertex] (x1) at (-1.8, -1.5) {};
                \node[cvertex] (x2) at (-2, -.5) {};
                \node[cvertex] (x3) at (-1.6, -1) {};
                \node[fit=(x1)(x2)(x3), rounded corners=3pt, draw=lipicsGray!40,
                    fill=lipicsGray!10] {};
                \node[fit=(e1)(e2)(f1)(f2), rounded corners=5pt, draw] (f) {};
                \node at ($(f.south) + (0,-.3)$) {$\hat{G}$};

                \node[cvertex] (x1) at (-1.7, -1.5) {};
                \node[cvertex] (x2) at (-2, -.5) {};
                \node[cvertex] (x3) at (-1.6, -.84) {};
                \node[cvertex] (x4) at (-2, -1.16) {};

                \node[vertex] (U3) at (2, -0) {};
                \node[vertex] (U4) at (.5, -0) {};
                \node[vertex] (V2) at (2, -2) {};
                \node[vertex] (V3) at (1.25, -2) {};
                \node[vertex] (V4) at (0.5, -2) {};
                \draw[cedge] (U4) -- (x3) -- (V4);
                \draw[cedge] (U3) -- (x3) -- (V2);
                \draw[cedge] (x3) -- (V3);
                \draw[cedge] (U3) -- (x2);
                \draw[cedge] (U4) -- (x2);
                \draw[cedge] (V4) -- (x1) -- (x2) -- (x3) -- (x4);
                \draw[cedge] (V2) -- (x1);
                \draw[cedge] (V3) -- (x1);
                \draw[sdedge] (U3) -- (V4);
                \draw[sdedge] (U3) -- (V2);
                \draw[sdedge] (V3) -- (U4) -- (V4);
            \end{scope}
        \end{tikzpicture}
        \caption{The construction of $\hat{G}$ from
        \cref{lem:reduction_hereditary}.}\label{fig:ghat}
    \end{figure}

	Note that the construction induces a partition of the vertices of $\hat{G}$ into three sets:
	\[V(\hat{G}) = R ~\dot\cup~ U ~\dot\cup~ V,\]
    where $R = V(H)\setminus \{u,v\}$; we set $r:= |R|$. Now define
    \[\indsubs{\Phi,k+r}{\hat{G}}[R] := \{ F \in \indsubs{\Phi,k+r}{\hat{G}} ~|~
    R\subseteq V(F)  \},\]
    that is, $\indsubs{\Phi,k+r}{\hat{G}}[R]$ is the set of all induced subgraphs $F$ of
    size $k+r$ in $\hat{G}$ that satisfy $\Phi$ and that contain all vertices in $R$.
    Next, we show that the cardinality of $\indsubs{\Phi,k+r}{\hat{G}}[R]$ reveals the
    number of independent sets of size $k$ in $G$.
	\begin{claim}
		Let $\mathsf{IS}_k$ denote the set of independent sets of size $k$ in $G$. We have
		\[\#\indsubs{\Phi,k+r}{\hat{G}}[R] =\#\mathsf{IS}_k. \]
	\end{claim}
    \begin{claimproof}
        Let $b$ denote the function that maps a graph $F \in \indsubs{\Phi,k+r}{\hat{G}}[R]$ to a
        $k$-vertex subset of~$G$ given by \[b(F) := V(F) \cap (U ~\dot\cup ~V).\] We
        show that $\mathsf{im}(b) = \mathsf{IS}_k$.

        ``$\subseteq$'': Fix a graph $F \in \indsubs{\Phi,k+r}{\hat{G}}[R]$ and
        write $b(F)=U' ~\dot\cup~V'$ where $U' \subseteq U$ and $V' \subseteq V$. As
        $|V(F)| = k+r$, as well as $V(F) = R~\dot\cup~ U' ~\dot\cup ~V'$, and $|R|=r$,
        we see that $|b(F)|=k$. Now assume that
        $b(F)$ is not an independent set, that is, there is an edge $(u,v)\in U' \times V'$ in
        $F$. Observing that the induced subgraph $F[R\cup
        \{u,v\}]$ of $F$ is isomorphic to $H$, and that $H$ is a forbidden induced subgraph of the
        property $\Phi$, yields the desired contradiction.

        ``$\supseteq$'': Let $U' ~\dot\cup ~V'$ denote an independent set of size $k$ of $G$
        with $|U'|=k_1$ and $|V'|=k_2$; note that $k_1$ or $k_2$ might be zero. Let $F$
        denote the induced subgraph of $\hat{G}$ with vertices $U' ~\dot\cup ~V'
        ~\dot\cup~ R$. Then $F$ has $k+r$ vertices and is isomorphic to $H_{u,v}^{k_1,k_2}$.
        Suppose $F$ does not satisfy $\Phi$. Then $F$ has an induced subgraph isomorphic to a
        graph in $\Gamma(\Phi)$. However, this is impossible as
        $\indsubs{\hat{H}}{H_{u,v}^{x,y}} =
        \emptyset$ for all $\hat{H}\in \Gamma(\Phi)$ and non-negative integers $x,y$.

        This shows that $b$ is a surjective function from $\indsubs{\Phi,k+r}{\hat{G}}[R]$ to
        $\mathsf{IS}_k$. Furthermore, injectivity is immediate by the definition of $b$ as
        $V(F)\setminus (U~\dot\cup ~V) =R$ for every $F \in \indsubs{\Phi,k+r}{\hat{G}}[R]$,
        which proves the claim.
    \end{claimproof}
    It hence remains to show how our algorithm $\mathbb{A}$ can compute the cardinality of
    $\indsubs{\Phi,k+r}{\hat{G}}[R]$.
    \begin{claim}
        Write $R=\{z_1,\dots,z_r\}$. We see that
        \[ \#\indsubs{\Phi,k+r}{\hat{G}}[R] = \sum_{J \subseteq[r]} (-1)^{|J|} \cdot
        \#\indsubs{\Phi,k+r}{\hat{G}\setminus J}, \]
        where $\hat{G}\setminus J$ is the graph obtained from $\hat{G}$ by deleting all
        vertices $z_i$ with $i\in J$.
    \end{claim}
    \begin{claimproof}
        Using the principle of inclusion and exclusion, we obtain that
        \begin{align*}
            &\#\indsubs{\Phi,k+r}{\hat{G}}[R]\\
            &\quad= \#\indsubs{\Phi,k+r}{\hat{G}} - \#\{F \in \indsubs{\Phi,k+r}{\hat{G}}
            ~|~ \exists i \in [r]: z_i \notin V(F) \} \\
            &\quad= \#\indsubs{\Phi,k+r}{\hat{G}} - \left|\bigcup_{i=1}^r \{F \in
            \indsubs{\Phi,k+r}{\hat{G}} ~|~ z_i \notin V(F) \} \right|\\
            &\quad= \#\indsubs{\Phi,k+r}{\hat{G}} - \sum_{\emptyset \neq J \subseteq [r]}
            (-1)^{|J|+1} \left| \bigcap_{i \in J} \{F \in \indsubs{\Phi,k+r}{\hat{G}} ~|~
            z_i \notin V(F) \}\right| \\
            &\quad= \#\indsubs{\Phi,k+r}{\hat{G}} - \sum_{\emptyset \neq J \subseteq [r]}
            (-1)^{|J|+1} \# \{F \in \indsubs{\Phi,k+r}{\hat{G}} ~|~ \forall i \in J: z_i
            \notin V(F) \} \\
            &\quad= \#\indsubs{\Phi,k+r}{\hat{G}} - \sum_{\emptyset \neq J \subseteq [r]}
            (-1)^{|J|+1} \#\indsubs{\Phi,k+r}{\hat{G}\setminus J} \\
            &\quad=\sum_{J \subseteq[r]} (-1)^{|J|} \cdot
            \#\indsubs{\Phi,k+r}{\hat{G}\setminus J}\,.
        \end{align*}
    \end{claimproof}
    Consequently, the algorithm $\mathbb{A}$ requires linear time in $|G|$ and $2^r\in
    O(1)$ oracle calls, each query of the form $(\hat{G}\setminus J, k+r)$, to compute the
    number of independent sets of size $k$ in $G$ as shown in the previous claims. In
    particular, $|V(\hat{G}\setminus J)| \in O(|V(G)|)$ for each $J \subseteq [r]$ and
    $k+r \in O(k)$; this completes the proof.
\end{proof}

\begin{theorem}\label{thm:critical_hardness}
    Let $\Phi$ denote a computable and critical hereditary graph property. Then
    $\#\indsubsprob(\Phi)$ is $\#\W{1}$-complete and cannot be solved in time
    \[g(k)\cdot |V(G)|^{o(k)}\] for any function $g$, unless ETH fails. The same is true for
    the problem $\#\indsubsprob(\neg\Phi)$.
\end{theorem}
\begin{proof}
    It is known that counting independent sets of size $k$ in bipartite graphs is
    $\#\W{1}$-hard~\cite{CurticapeanDFGL19} and cannot be solved in time $g(k)\cdot
    |V(G)|^{o(k)}$ for any function $g$, unless ETH fails~\cite{DorflerRSW19}. The theorem thus
    follows by \cref{lem:reduction_hereditary}.
\end{proof}
As a particular consequence we establish a complete classification for the properties of
being $H$-free, including for instance claw-free graphs and co-graphs~\cite{Seinsche74}.

\hermn*
\begin{proof}
	Holds by \cref{lem:hfree_critical,thm:critical_hardness}.
\end{proof}
Finally, note that the previous result is in sharp contrast to the result of Khot and
Raman~\cite{KhotR02} concerning the decision version of $\#\indsubsprob(\Phi)$: Their
result implies that \emph{finding} an induced subgraph of size $k$ in a graph $G$ that
satisfies a hereditary property $\Phi$ with $\Gamma(\Phi)=\{H\}$ for some graph $H$ can be
done in time $f(k)\cdot|V(G)|^{O(1)}$ for some function $f$ if $H$ is neither a clique nor
an independent set of size at least $2$. \newpage

\bibliography{conference}
\end{document}